\documentclass[a4, 11pt, reqno, final]{amsart}

\usepackage[margin=1.0in]{geometry}
\usepackage{amssymb,  mathtools}
\usepackage[colorlinks=true, urlcolor=blue, citecolor=blue, linkcolor=blue, raiselinks=true, breaklinks=true]{hyperref}
\usepackage{enumerate}
\usepackage[square, numbers, sort&compress]{natbib}
\usepackage{doi}
\usepackage{tikz}
\usetikzlibrary{positioning, patterns}
\usepackage{tikz-3dplot}
\usepackage[notref,notcite]{showkeys}
\usepackage{tikz}
\usepackage{pgfplots}
\usepackage{lineno}

\newcommand{\mres}{\mathbin{\vrule height 1.6ex depth 0pt width
		0.13ex\vrule height 0.13ex depth 0pt width 1.3ex}}

\newcommand{\calA}{\mathcal{A}}
\newcommand{\calB}{\mathcal{B}}
\newcommand{\calD}{\mathcal{D}}

\newcommand{\calE}{\mathcal{E}}

\newcommand{\calL}{\mathcal{L}}
\newcommand{\calM}{\mathcal{M}}

\newcommand{\calX}{\mathcal{X}}
\newcommand{\bbR}{\mathbb{R}}

\newcommand{\n}[1]{\left\|#1\right\|}  
\newcommand\embed{\hookrightarrow} 
\newcommand\e{\varepsilon}
\newcommand\weakstar{\stackrel{\ast }{\rightharpoonup}}  
\newcommand\weak{\rightharpoonup} 

\newcommand\ringring[1]{%
	{
		\mathop{\kern0pt #1}\limits^{
			\vbox to-1.85ex{
				\kern-2ex 
				\hbox to 0pt{\hss\normalfont\kern.1em \r{}\kern-.45em \r{}\hss}%
				\vss 
			}
		}
	}
}

\theoremstyle{plain}
\newtheorem{theorem}{Theorem}[section]
\newtheorem{lemma}[theorem]{Lemma}
\newtheorem{corollary}[theorem]{Corollary}
\newtheorem{remark}[theorem]{Remark}

{}

\DeclareMathOperator{\st}{\,:\,}

\DeclareMathOperator{\cof}{cof\,}
\DeclareMathOperator{\curl}{curl\,}
\renewcommand{\div}{\mathrm{div}\,}
\DeclareMathOperator{\Sym}{\bbR^{2\times 2}_{\rm sym}}

\newcommand{\Rint}{\mathop{\mathrlap{\pushpv}}\!\int} 
\newcommand{\pushpv}{\mathchoice
	{\mkern5mu\rule[.6ex]{.5em}{1pt}}
	{\mkern2.8mu\rule[.5ex]{.35em}{.8pt}}
	{\mkern2.5mu\rule[.29ex]{.3em}{.7pt}}
	{\mkern2mu\rule[.2ex]{.2em}{.5pt}}}

\newcommand*{\tran}{^{\mkern-1.5mu\mathsf{T}}}

\begin{document}
	\sloppy

	\title[Two one-dimensional models for transversely curved shallow shells]{Derivations of two one-dimensional models for transversely curved shallow shells: \\ one leads to relaxation}

	\author[R. Paroni]{Roberto Paroni}
	\address[R. Paroni]{Department of Civil and Industrial Engineering, University of Pisa, Largo Lucio Lazzarino 1, 56122, Pisa, Italy}
	\email{roberto.paroni@unipi.it}

	\author[M. Picchi Scardaoni]{Marco Picchi Scardaoni}
	\address[M. Picchi Scardaoni]{Department of Civil and Industrial Engineering, University of Pisa, Largo Lucio Lazzarino 1, 56122, Pisa, Italy}
	\email{marco.picchiscardaoni@ing.unipi.it}

	\date{\today}

	\begin{abstract}
		We study the $\Gamma$-limit of sequences of variational problems for straight, transversely curved shallow shells, as the width of the planform $\e$ goes to zero.
		The energy is of von Kármán type for shallow shells under suitable boundary conditions. What distinguishes the various regimes is the scaling of the stretching energy $\sim \e^{2\beta}$, with $\beta$ a positive number. We derive two one-dimensional models as $\beta$ ranges in $(0, 2]$. Remarkably, boundary conditions are essential to get compactness. \\
		We show that for $\beta \in (0, 2)$ the $\Gamma$-limit leads to relaxation: the limit membrane energy vanishes on compression. For $\beta=2$ there is no relaxation, and the limit model is a nonlinear energy coupling four kinematical descriptors in a nontrivial way.
		As special cases of the latter limit model, a nonlinear Vlasov torsion theory and a nonlinear Euler-Bernoulli beam theory can be deduced.

	\end{abstract}

	\maketitle

	\tableofcontents

	\section{Introduction}

	The rigorous derivation of one-dimensional mechanical theories from higher-dimensional frameworks has seen significant development in recent years.

	In this work, we focus on straight, transversely curved ribbons in their natural configuration. That is, we consider transversely curved shallow shells whose planform projection is a rectangle. In this setting, the rigorous derivation of limiting one-dimensional models, as the width of the ribbon planform tends to zero, was first addressed in~\cite{Paroni2024c} (see also~\cite{Coppede2024a}), inspired by the bending behavior of a carpenter’s tape measure, where deformation can localize. The resulting limit model captures curvature concentration (i.e., the formation of elastic hinges). However, the starting point in that work was a relatively strong energy, whereas more realistic models, such as those based on von Kármán energies, have weaker compactness properties, posing additional challenges.

	For completeness, we note that several one-dimensional models have been heuristically proposed in the literature for transversely curved shallow shells; see, e.g.,~\cite{Brunetti2020, Guinot2012, Picault2016, Kumar2023, Mao2024}.
	Rigorous derivations of one-dimensional models for ribbons with a curved cross-section starting from three-dimensional nonlinear elasticity can be found in~\cite{Davoli2013, Maor2025}, while the case of flat cross-sections has been addressed in~\cite{Freddi2012, Freddi2013}.
	Derivations starting instead from two-dimensional theories are presented in~\cite{Freddi2016a, Freddi2018}.

	Here, we consider families of von Kármán-type energies, as commonly used for modeling shallow shells and prestrained plates~\cite{Brunetti2020a, Brunetti2020c, Fabbrini2025, Kloda2025, JimenezBolanos2021, Lewicka2017, Tobasco2021a, Bella2025, Bartels2025, Velcic2012, PicchiScardaoni2021d}. In particular, we impose a specific scaling for the stretching energy: as the ribbon width~$\e$ tends to zero, we assume the natural curvature scales like~$1/\e$, and the stretching energy is of order~$\e^{2\beta}$, with~$\beta$ a non-negative real number.

	Additionally, we assume that one short end of the shell is clamped, while the other undergoes a rigid roto-translation at a suitable scale. As it will become evident, the order of the longitudinal displacement plays a crucial role in determining the limit behavior. Specifically, we obtain compactness for appropriately scaled displacement sequences, of order~$\e^\beta$ for in-plane displacements and~$\e^{\beta/2}$ for out-of-plane ones. Furthermore, for~$\beta \in (0, 2)$, the problem exhibits relaxation. In contrast, when~$\beta = 2$, the energy functional is sufficiently rigid to prevent relaxation. In all cases, a shared feature is compactness in relatively weak topologies, alongside surprisingly regular limit displacements.

	It is well known that boundary conditions can significantly influence the resulting limit models. In the context of three-dimensional nonlinear elasticity, for instance, clamping a plate along its boundary can (under suitable scalings) lead to a membrane theory whose energy vanishes on contractions~\cite{Conti2005}. A broader overview of how various scaling regimes interact with boundary conditions is presented in~\cite{Friesecke2006}.

	In our setting, a similar phenomenon occurs, arguably in an even more evident way. When natural boundary conditions are imposed, a wide class of displacements yields zero membrane energy, resulting in a soft response dominated by out-of-plane bending. However, for~$\beta \in (0,2)$, introducing the additional boundary condition stiffens the shell considerably: the kernel of the membrane energy is reduced or even eliminated. The practical consequence is that the out-of-plane bending term becomes negligible in the limit model (i.e., the second derivative of the transversal displacement disappears).

	The case~$\beta = 2$ is critical: both the membrane and bending contributions are retained, and a ``complete'' limit theory is obtained. In this regime, the four kinematic descriptors of the limit displacement are nontrivially coupled. Interestingly, it also includes, as a particular case, a nonlinear version of Vlasov’s torsion theory.

	To the best of our knowledge, the two limit models obtained in the present work are new.

	The paper is structured as follows. Section~\ref{sec:prel} reviews the main technical tools and establishes notation. In Section~\ref{sec:prob}, we introduce the sequence of variational problems under consideration. Section~\ref{sec:comp} is devoted to compactness results. Section~\ref{sec:gamma2} proves the $\Gamma$-convergence result in the case~$\beta = 2$, while Section~\ref{sec:gamma1} addresses the case~$\beta < 2$.

	\section{Preliminaries and notation}\label{sec:prel}

	The Euclidean (Frobenius) product in $\mathbb{R}^N$ is indicated with $\cdot$ and the corresponding induced norm by $|\cdot|$. If $x=(x_1,x_2) \in \bbR^2$, $x^\perp \coloneqq   (-x_2, x_1)$.\\
	Let us summarize some useful properties of matrices. For every $A,B \in \Sym$:
	\begin{equation}\label{eq:detAB}
		\begin{aligned}
			\det(A\pm B) &= \det A + \det B \pm  A \cdot \cof B,\\
			|A|^2 &\geq 2|\det A|,\\
			A\cdot\cof A &= 2\det A.
		\end{aligned}
	\end{equation}
	where
	$(\cof A)_{\alpha\beta} \coloneqq   \calE_{\alpha\gamma}\calE_{\beta\delta} A_{\gamma\delta}$ and $\calE$ is the Levi-Civita symbol.
	The dyadic product of two vectors $a,b \in \bbR^n$, written as  $a \otimes b$, is defined by $(a \otimes b)(c) = (c \cdot b) a$ for every $c \in \bbR^n$.
	If not specified, we adopt Einstein' summation convention for indices, and $C$ denotes a positive constant that may vary from line to line.\\
	We denote the integral average by $\Rint_\Omega f \, dx \coloneqq  \frac{1}{|\Omega|}\int_\Omega f \, dx$.
	If $f,g:\bbR \to \bbR$, their convolution $f\ast g$ is given by $(f\ast g)(x) \coloneqq  \int_\bbR f(x-y)g(y) \, dy = \int_\bbR f(y)g(x-y) \, dy$ (whenever the integrals exist).

	The characteristic function of the set $A\subset \bbR^N$ is denoted by
	$$ \mathbf{1}_A(x) \coloneqq   \begin{cases}
		1 & x \in A\\
		0 & x \in \bbR^N \setminus A.
	\end{cases}$$
	The positive part of $a$ is denoted by $a^+ \coloneqq  \max\{a,0\}$ while the negative part of $a$ is $a^- \coloneqq -\min\{a,0\}$, so that $a = a^+ - a^-$.

	Let $\Omega\subset \mathbb{R}^N$ be an open, bounded, Lipschitz domain. If $k\in \mathbb{N}\cup \{\infty\}$, then $C^k(\Omega)$ denotes the space of real-valued, $k$-times continuously differentiable functions on $\Omega$ and $C^k(\overline{\Omega})$ denotes the space of real-valued, $k$-times continuously differentiable functions up to the boundary of $\Omega$. $C_0^k(\Omega)$ denotes the completion with respect to the sup-norm of $C_c^k(\Omega)$, the space of functions belonging to $C^k(\Omega)$ that have compact support in $\Omega$. The dual space of infinitely differentiable functions with compact support on $\Omega$ (space of distributions) is denoted by $\calD'(\Omega)$.

	The Lebesgue spaces of $p$-integrable functions on $\Omega$ are denoted by $L^p(\Omega)$ for $1\leq p\leq \infty$.
	We endow $L^2(\Omega)$ with the canonical scalar product $(f,g) \mapsto \int_\Omega fg \, dx$ for every $f,g \in L^2(\Omega)$.\\
	The Sobolev' spaces of $L^p(\Omega)$ functions whose derivatives up to the order $k$ are in $L^p(\Omega)$ are denoted by $W^{k,p}(\Omega)$.
	We denote strong convergence (convergence in norm) with the symbol $\to$, while weak convergence will be denoted by $\weak$.

	The strain of a vector-valued function $u \in H^1(\Omega, \bbR^2)$ is defined as $$Eu = \frac12(\nabla u\tran + \nabla u).$$

	Let $H^2_0(\Omega)$ be the closure of $C_c^\infty(\Omega)$ with respect to the $H^2$ norm.
	The dual space of $H_0^2(\Omega)$ is denoted by $H^{-2}(\Omega)$.
	The latter is endowed with the usual dual norm
	$$ \n{v}_{H^{-2}(\Omega)} = \sup\limits_{\varphi\in H^2_0(\Omega),\\ \n{\varphi}_{H^2(\Omega)=1}} |\langle v, \varphi \rangle|.$$

	We introduce the operator $\curl\curl: L^2(\Omega, \Sym) \to H^{-2}(\Omega)$ defined as $$\curl\curl A = \partial_{11}A_{22} + \partial_{22}A_{11} - 2\partial_{12}A_{12} \qquad \forall \  A \in L^2(\Omega, \Sym).$$
	We recall that $\curl\curl$  and $\cof\nabla^2: H_0^2(\Omega) \to L^2(\Omega,\Sym)$ are adjoint operators, in the sense that 	\begin{equation}\label{ccA}\langle \curl\curl A, \varphi\rangle = \int_\Omega A\cdot \cof\nabla^2 \varphi \ dx \qquad \forall \ A \in L^2(\Omega, \Sym), \ \forall \ \varphi \in H_0^2(\Omega).\end{equation}
	In particular,  we have in $H^{-2}$ (see for instance \cite{DOnofrio2005}) that
	\begin{equation}\label{cceu}
		\curl\curl(E u) = 0 \qquad \forall \ u \in H^1(\Omega, \bbR^2)
	\end{equation}
	and 	that
	\begin{equation}\label{ccww}\curl\curl(\nabla w \otimes \nabla w) = -2\det\nabla^2 w \qquad \forall \ w \in H^2(\Omega).
	\end{equation}

	We denote by $\calM(\Omega)$ the space of finite Radon measures on $\Omega$, and by $\calM^+(\Omega)$ the subset of the non-negative ones.
	The restriction of a measure $\mu$ on $\bbR^N$ to a measurable set $E\subset \bbR^N$ is the measure $\mu\mres E$ defined as $\mu\mres E (F) \coloneqq   \mu(F \cap E)$ for all measurable sets $F\subset \bbR^N$.
	Let $A, B$ be two sets and let $\mu \in \calM(A)$, $\nu \in \calM(B)$. The product measure
	$\mu \otimes \nu \in \calM(A\times B)$  is the measure satisfying $(\mu \otimes \nu) (E\times F)$ = $\mu(E)\nu(F)$ for every Borel set $E \in A$ and $F \in B$.\\
	With $\calL^n$ we denote the $n$-dimensional Lebesgue measure.
	If $I$ is an interval and $\mu \in \calM(I)$, by the Radon-Nikodym decomposition we can write uniquely $\mu = \frac{d\mu}{d\calL}\calL + \mu_s$ where $\frac{d\mu}{d\calL}\calL$ is the part of $\mu$ which is absolutely continuous with respect to $\calL$, $\frac{d\mu}{d\calL}$ is the Radon-Nikodym derivative of $\mu$ with respect to $\calL$, and $\mu_s$ is the singular part of $\mu$ (with respect to $\calL$).

	Indicating by $Du$ the distributional derivative of $u$, the spaces of functions of bounded variation, bounded Hessian, and bounded deformation are defined as
	\begin{equation}\label{eq:bh}
		\begin{aligned}
			BV(\Omega) &\coloneqq   \{u \in L^{1}(\Omega) : Du \in \calM(\Omega,\bbR^N)\},\\
			BH(\Omega) &\coloneqq   \{u \in W^{1,1}(\Omega) : D^2u \in \calM(\Omega,\bbR^{N\times N}_{\mathrm{sym}})\},\\
			BD(\Omega, \bbR^2) &\coloneqq   \{u \in L^{1}(\Omega, \bbR^2) : \frac12(Du + Du\tran) \in \calM(\Omega,\Sym)\}.
		\end{aligned}
	\end{equation}

	The (total) variation measure of a function in $u \in BV(\Omega)$
	is defined as
	$$
	|Du|(\Omega) \coloneqq   \sup \{\int_{\Omega} u\,\div\varphi \,dx \st \varphi \in C^1_c (\Omega,\bbR^N), \ \n{\varphi}_{\infty} \leq 1 \}.
	$$
	According to the Riesz's representation theorem, $\calM(\Omega)$ can be identified with the topological dual of $C_0(\Omega)$.
	We say that $(\mu_n) \subset \calM(\Omega)$ converges weakly$^\ast $ to $\mu\in\calM(\Omega)$ (and we write $\mu_n \weakstar \mu$ in $\calM(\Omega)$) if for every $\varphi\in C_0(\Omega)$ $\lim\limits_{n\uparrow \infty}\int_\Omega \varphi\, d\mu_n = \int_\Omega \varphi \,d\mu$.\\
	We say $(u_n) \subset BD(\Omega, \bbR^2)$ converges weakly$^\ast $ in $BD(\Omega, \bbR^2)$ to $u\in BD(\Omega, \bbR^2)$ (and we write $u_n \weakstar u$ in $BD(\Omega, \bbR^2)$) if $u_n\to u$ in $L^{1}(\Omega, \bbR^2)$ and $\frac{Du_n + Du_n\tran}{2} \weakstar \frac{Du + Du\tran}{2}$ in $\calM(\Omega, \Sym)$.\\
	We further recall
	the continuous embedding $BD(\Omega, \bbR^2)\embed L^2(\Omega, \bbR^2)$ (see \cite{Temam1987} for details).

	If $(a_\e)_i$ is a component of an $\e$-parameterized vector- (or matrix-) valued sequence, from time to time we may write equivalently $(a_{\e i})$ or $(a^\e_i)$.

	\section{The Problem}\label{sec:prob}

	Let $\e$ be a sequence of positive numbers converging to zero. We consider a family of shells, parameterized by $\e$, obtained by translating a planar curve (the \emph{cross-section}) lying in the $x_2$--$x_3$ plane along the $x_1$-axis.

	Let $\ell > 0$, define the domains
	$$
	I \coloneqq \left(-\frac{\ell}{2}, \frac{\ell}{2}\right), \qquad W_\e \coloneqq \left(-\frac{\e}{2}, \frac{\e}{2}\right), \qquad \Omega_\e \coloneqq I \times W_\e.
	$$
	The region $\Omega_\e$ is the projection of the shell onto the $x_1$--$x_2$ plane.

	\noindent
	We also define
	$$
	W \coloneqq W_1, \qquad \Omega \coloneqq \Omega_1 = I \times W.
	$$

	To describe the cross-section, we consider a non-affine, even function $\mathring{w} \in C^2(\overline{W})$ satisfying
	\begin{equation}\label{avwe0}
		\int_W \mathring{w}(x_2) \, dx_2 = 0
	\end{equation}
	and such that $\mathring{w}'' = 0$ only in a finite number of points of the interior part of $W$.
	Let $Y$ be the set of points in $W$ in which the curvature of the cross-section is zero:
	\begin{equation}\label{defY}
		Y := \{ y \in W \st \mathring{w}''(y) = 0\},
	\end{equation}
	and
	$$ I\times Y  = \{(x_1,x_2) \in \Omega \st x_2 \in Y \}.$$
	Whenever $Y$ is not nonempty, $I\times Y $ is the union in $\Omega$ of a finite number of disjoint segments parallel to $\mathsf{e}_1$ of length $\ell$.
	We define the shell's natural configuration $\mathring{v}_\e \in C^2(\overline{W}_\e)$ as
	$$
	\mathring{v}_\e(x_2) \coloneqq \e \mathring{w}(\frac{x_2}\e).
	$$
	Then, the cross-section of the shell parameterized by $\e$ is described by the curve $x_2 \mapsto (x_2, \mathring{v}_\e(x_2))$ in the $x_2$--$x_3$ plane.

	\begin{figure}[!hbt]
		\centering
		\def\L{5}
		\def\w{1}
		\tdplotsetmaincoords{60}{-135}
		\begin{tikzpicture}[tdplot_main_coords, scale=1, transform shape]
			\coordinate (O) at (0,0,0);
			\coordinate (A) at (0,-2*\w,0);
			\coordinate (B) at (0,0,2*\w);
			\coordinate (C) at (-2*\w,0,0);
			\coordinate (D) at (-\L,\w,0.);
			\coordinate (E) at (\L,\w,0.);
			\coordinate (F) at (\L,-\w,0.);
			\coordinate (G) at (-\L,-\w,0.);

			\draw[red] (\L,0,0) node[above] {$\Omega_\e$};
			\draw[blue] (F)+(-\w,0,\w) node[above] {$\mathring{v}_\e$};

			\begin{scope}[xshift=0.06cm,yshift=0.13cm]
				\pgfmathsetmacro{\step}{\L/100}
				\draw[gray,thick,fill=red,opacity=0.3] (D) -- (E) -- (F) -- (G) -- cycle;
				\draw[black,thick] (D) -- (E) -- (F) -- (G) -- cycle;
				\foreach \y in {0,\step,...,10}{
					\draw[cyan, thick,opacity=0.5] plot[domain=-\w:\w,smooth,variable=\t] ({ -\L+\y},{\t},{(1.5*\t*\t/\w/\w-\w/2)^2 - 4/10} );
				}
				\draw[blue, thick] plot[domain=-\w:\w,smooth,variable=\t] ({ -\L},{\t},{(1.5*\t*\t/\w/\w-\w/2)^2 - 4/10});
				\draw[blue, thick] plot[domain=-\w:\w,smooth,variable=\t] ({ -\L+10},{\t},{(1.5*\t*\t/\w/\w-\w/2)^2 - 4/10});
				\draw[blue, thick] plot[domain=-\L:\L,smooth,variable=\t] ({ \t},{-\w},{(1.5-\w/2)^2 - 4/10});
				\draw[blue, thick] plot[domain=-\L:\L,smooth,variable=\t] ({ \t},{\w},{(1.5-\w/2)^2 - 4/10});
			\end{scope}

			\draw[-latex] (O) -- (B) node[above] {$x_3$};
			\draw[-latex] (O) -- (C) node[above] {$x_1$};
			\draw[-latex] (O) -- (A) node[right] {$x_2$};
		\end{tikzpicture}
		\caption{Reference configuration and shell planform}
		\label{fig:1}
	\end{figure}
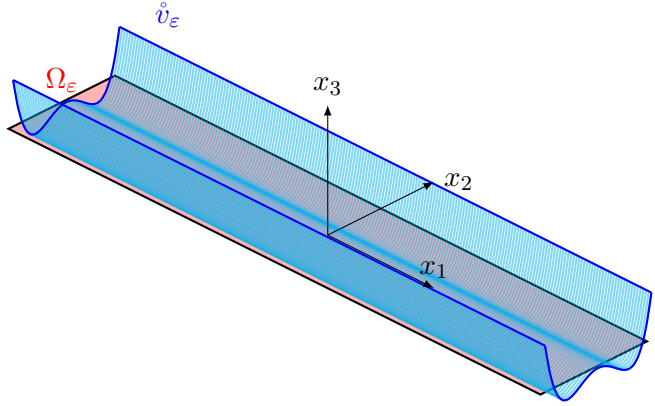

	Let $\beta \in (0, 2]$. For an out-of-plane displacement $v \in H^2(\Omega_\e)$ (in the $x_3$ direction) and an in-plane displacement $y \in H^1(\Omega_\e; \bbR^2)$, we define the von Kármán-type energy functional
	\begin{equation}\label{Ie}
		F_\e^\beta(y,v) \coloneqq \int_{\Omega_\e} \frac{1}{2}|\nabla^2 v - \nabla^2 \mathring{v}_\e|^2 + \frac{1}{2} \frac{1}{\e^{2\beta}} \left|Ey + \frac{1}{2} \nabla v \otimes \nabla v - \frac{1}{2} \nabla \mathring{v}_\e \otimes \nabla \mathring{v}_\e \right|^2 \, dx,
	\end{equation}
	subject to the following set of boundary conditions
	\begin{equation}\label{condpresc}
		\begin{aligned}
			& y(-\frac{\ell}{2}, x_2) = (0,0), \qquad v(-\frac{\ell}{2}, x_2) = \mathring{v}_\e(x_2), \qquad \partial_1 v(-\frac{\ell}{2}, x_2)=0, \\
			&\begin{multlined}[t]y(\frac{\ell}{2}, x_2) =
			(\e^\beta\Lambda_1 -\e^{\beta/2}\Phi_2\mathring{v}_\e(x_2) -\e^{\beta/2} \Phi_2\Phi_1 x_2 - \e^{\beta-1}\Phi_3 x_2, \\   \e^{\beta-1} \Lambda_2 - \Phi_1\mathring{v}_\e(x_2) -\frac{\Phi_1^2 x_2}{2}),\end{multlined}\\
			&v(\frac{\ell}{2}, x_2) = \mathring{v}_\e(x_2) + \e^{\beta/2} \Lambda_3 + x_2 \Phi_1, \qquad   \partial_1 v(\frac{\ell}{2}, x_2)=\e^{\beta/2} \Phi_2.
		\end{aligned}
	\end{equation}
	for a.e. $x_2 \in \e(-\frac12, \frac12)$ and
	for given real constants $\Lambda_1, \Lambda_2, \Lambda_3, \Phi_1, \Phi_2, \Phi_3$.
	The choice of the exponents of $\e$ appearing in \eqref{condpresc} will become transparent in Section \ref{sec:comp}.\\
	Mechanically, we may think the thickness of the shell to scale like $\e^\beta$. The further condition $\beta \ge 1$ ensures the thickness vanishes at least as fast as the width $\e$. However, from a mathematical viewpoint, we need not to restrict $\beta$ to the interval $[1,2]$.

	The boundary conditions \eqref{condpresc} ensure that the end of the shell at $x_1 = - \frac{\ell}{2}$ is clamped. The constants $\Lambda_i$ can be interpreted as the average translation of the cross section at $x_1=\ell/2$ along $\mathsf{e}_i$, while the constants $\Phi_i$ represent rotations (slopes) of the endmost section around $\mathsf{e}_i$.

	We
	introduce the scaled out-of-plane displacement $w: \Omega\to\bbR$
	$$
	w(x_1, x_2) \coloneqq   v(x_1, \e x_2),
	$$
	the scaled cross-section $\mathring{w}_\e: \Omega\to\bbR$
	$$
	\mathring{w}_\e(x_1, x_2) \coloneqq   \mathring{v}_\e(\e x_2)=\e \mathring{w}(x_2),
	$$
	and
	the scaled in-plane displacements $u: \Omega\to\bbR^2$
	$$
	u_1(x_1, x_2) \coloneqq   y_1(x_1, \e x_2), \qquad u_2(x_1, x_2) \coloneqq   \e y_2(x_1, \e x_2).
	$$
	We also define the scaled operators as
	$$
	\nabla_\e z\coloneqq   \left(\partial_1 z, \e^{-1}\partial_2 z\right)\tran, \qquad  \nabla^2_\e z\coloneqq
	\begin{pmatrix}
		\partial_{11}z & \e^{-1}\partial_{12}z\\
		\e^{-1}\partial_{21}z & \e^{-2}\partial_{22}z
	\end{pmatrix},
	$$
	$$
	E_\e z\coloneqq
	\begin{pmatrix}
		(Ez)_{11} & \e^{-1}(Ez)_{12}\\
		\e^{-1}(Ez)_{21} & \e^{-2}(Ez)_{22}
	\end{pmatrix}.
	$$
	Changing variables in \eqref{Ie} and \eqref{condpresc}, and dividing by the Jacobian, we obtain the rescaled energy
	\begin{align}\label{Fe}
		\hat{F}_\e^\beta(u, w) \coloneqq & \frac{1}{2}\int_\Omega  \left| \nabla^2_\e w - \frac{1}{\e} \mathring{w}'' \, \mathsf{e}_2 \otimes \mathsf{e}_2 \right|^2 \ dx \nonumber\\ & \qquad + \frac{1}{2}\int_\Omega  \frac{1}{\e^{2\beta}} \left| E_\e u + \frac{1}{2} \nabla_\e w \otimes \nabla_\e w - \frac{1}{2} \mathring{w}'^2 \, \mathsf{e}_2 \otimes \mathsf{e}_2 \right|^2 \, dx,
	\end{align}
	under the constraints
	\begin{equation}\label{cond}
		\begin{aligned}
			&u(-\frac{\ell}{2}, x_2) = (0,0), \qquad w(-\frac{\ell}{2}, x_2) = \e\mathring{w}, \qquad \partial_1 w(-\frac{\ell}{2}, \cdot)=0, \\
			& 	u(\frac{\ell}{2}, x_2) =
			(\e^\beta\Lambda_1 -\e^{1+\beta/2}\Phi_2\mathring{w} -\e^{1+\beta/2} \Phi_2\Phi_1 x_2 - \e^\beta \Phi_3 x_2, \\ & \hspace{7cm} \e^\beta \Lambda_2 - \e^2 \Phi_1\mathring{w} -\e^2 \frac{\Phi_1^2 x_2}{2}),\\
			&	 w(\frac{\ell}{2}, x_2) = \e\mathring{w} + \e^{\beta/2} \Lambda_3 + \e x_2 \Phi_1, \qquad   \partial_1 w(\frac{\ell}{2}, x_2)=\e^{\beta/2} \Phi_2
		\end{aligned}
	\end{equation}
	valid for a.e. $x_2 \in (-\frac12, \frac12)$.
	Note that, in particular, the first and the fourth prescription in \eqref{cond} imply that
	\begin{equation}\label{condvar}
		\int_\Omega \partial_1 u_1(x_1, x_2) \, dx  = \Lambda_1 \e^\beta \quad \forall \e > 0.
	\end{equation}
	This constraint has a pivotal role in Section \ref{sec:comp}.

	With the function space $\calX \coloneqq L^1(\Omega; \bbR^2) \times L^2(\Omega)$ we define the extended functional on $\calX$:
	\begin{equation}\label{eq:augm}
		F_{\e}^\beta(u, w) \coloneqq
		\begin{cases}
			\hat{F}_\e^\beta(u, w) & \text{if } (u, w) \in H^1(\Omega; \bbR^2) \times H^2(\Omega) \text{ and (\ref{cond}) holds}, \\
			+\infty & \text{otherwise}.
		\end{cases}
	\end{equation}

	\subsection{Geometric Quantities}

	It is helpful to introduce the following constants, which depend on the domain $\Omega$:
	\begin{equation}\label{constants}
		\begin{aligned}
			J_0 &= \int_\Omega |x|^2 \, dx, &
			J_1 &= \int_W x_2^2 \, dx_2 = \frac{1}{12}, &
			J_2 &= \int_W \left(\frac{x_2^2}{2} - \frac{1}{24} \right)^2 \, dx_2 = \frac{1}{720}.
		\end{aligned}
	\end{equation}

	We also define two functions that depend on $\mathring{w}$:
	\begin{equation}\label{w0}
		\ringring{w}(x_2) \coloneqq \int_{0}^{x_2} \left(t \mathring{w}'(t) - \mathring{w}(t)\right) \, dt,
		\qquad
		\ringring{w}_{\langle 0 \rangle}(x_2) \coloneqq \ringring{w}(x_2) - c_1 x_2,
	\end{equation}
	where
	\begin{equation}\label{cJ0}
		c_1 = \frac{1}{J_1} \int_W s \ringring{w}(s) \, ds.
	\end{equation}

	Additionally, the following constants—dependent on the cross-section profile—will be useful:
	\begin{equation}\label{cJ}
		\begin{aligned}
			c_2 &\coloneqq \frac{1}{J_2} \int_W \left( \frac{t^2}{2} - \frac{1}{24} \right) \mathring{w}(t) \, dt, \\
			J_3 &\coloneqq \int_W \left( \mathring{w}(t) - c_2 \left( \frac{t^2}{2} - \frac{1}{24} \right) \right)^2 \, dt, &
			J_4 \coloneqq \int_W \ringring{w}_{\langle 0 \rangle}^2(t) \, dt.
		\end{aligned}
	\end{equation}

	It is noteworthy that $J_3 = 0$ if and only if
	$$
	\mathring{w} = \kappa \left( \frac{x_2^2}{2} - \frac{1}{24} \right), \qquad \text{for some } \kappa \in \mathbb{R}
	$$
	(in this case, $c_2 = \kappa$). Moreover, $c_2$ and $J_3$ cannot both be zero. Also, $J_4 = 0$ if and only if $\mathring{w}$ is affine, which is excluded by assumption. In fact, $J_4 = 0$ if and only if $\ringring{w}_{\langle 0 \rangle} = 0$, i.e.,
	$$
	\int_{0}^{x_2} \left(t \mathring{w}'(t) - \mathring{w}(t)\right) \, dt = c_1 x_2,
	$$
	which implies $t \mathring{w}'(t) - \mathring{w}(t) = c_1$. The general solution of this ordinary differential equation is
	$\mathring{w}(t) = at - c_1,$ with $a \in \mathbb{R}.$

	We conclude this section by stating some useful properties of $\ringring{w}$ and $\ringring{w}_{\langle 0 \rangle}$.

	\begin{lemma}\label{lemmaw0}
		Let $\ringring{w}$ and $\ringring{w}_{\langle 0 \rangle}$ be defined as in \eqref{w0}. Then:
		\begin{enumerate}[(i)]
			\item $\ringring{w}$ is odd;
			\item $\ringring{w}_{\langle 0 \rangle}$ and $\mathring{w} \, \ringring{w}_{\langle 0 \rangle}$ are odd;
			\item $\displaystyle \int_W t \, \ringring{w}_{\langle 0 \rangle}(t) \, dt = 0$;
			\item The functions $1$, $x_2$, $\mathring{w}(x_2)$, and $\ringring{w}_{\langle 0 \rangle}(x_2)$ are mutually orthogonal in $L^2(W)$;
			\item The functions $1$, $x_2$, $\displaystyle \frac{x_2^2}{2} - \frac{1}{24}$, $\mathring{w}(x_2) - c_2 \left( \frac{x_2^2}{2} - \frac{1}{24} \right)$, and $\ringring{w}_{\langle 0 \rangle}(x_2)$ are mutually orthogonal in $L^2(W)$.
		\end{enumerate}
	\end{lemma}

	\begin{proof}
		Integration by parts yields:
		$$
		\ringring{w}(x_2) = x_2 \mathring{w}(x_2) - 2 \int_0^{x_2} \mathring{w}(t) \, dt.
		$$
		Since $\mathring{w}$ is even, both $x_2 \mapsto x_2 \mathring{w}(x_2)$ and $x_2 \mapsto \int_0^{x_2} \mathring{w}(t) \, dt$ are odd, hence $\ringring{w}$ is odd. Therefore, $\ringring{w}_{\langle 0 \rangle}$ is also odd. Moreover, since $\mathring{w}$ is even, the product $\mathring{w} \, \ringring{w}_{\langle 0 \rangle}$ is odd.

		For the third point, using the definition of $c_1$:
		$$
		\int_W t \ringring{w}_{\langle 0 \rangle}(t) \, dt = \int_W t \ringring{w}(t) \, dt - c_1 \int_W t^2 \, dt = 0.
		$$

		The fourth claim follows from points (i)–(iii) and from \eqref{avwe0}.

		For the fifth point, the first three functions are orthogonal in $L^2(W)$ as they correspond to Legendre polynomials on $W$. By point (iv), $1$, $x_2$, $\mathring{w}(x_2)$, and $\ringring{w}_{\langle 0 \rangle}(x_2)$ are mutually orthogonal. Furthermore, since $\ringring{w}_{\langle 0 \rangle}$ is odd, it is orthogonal to the even function $\frac{x_2^2}{2} - \frac{1}{24}$. The result follows.
	\end{proof}

	\section{Compactness}\label{sec:comp}

	Throughout the section we consider a sequence  $(u_\e, w_{\e})\subset \calX$ such that
	\begin{equation}\label{cmp00}
		\sup_{\e} F_{\e}^{\beta}(u_\e, w_{\e}) <\infty.
	\end{equation}
	This bound implies that  $\sup_{\e} \hat{F}_{\e}^{\beta}(u_\e,w_\e)<+\infty$, where
	$$	 \begin{aligned}\hat{F}_{\e}^{\beta}(u_\e,w_\e) = & \frac12\int_\Omega  \Big| \nabla^2_\e w_\e- \frac{\mathring{w} ''}{\e} \mathsf{e}_2 \otimes \mathsf{e}_2\Big|^2 \, dx \\ & \hspace{2cm}+ \frac12\int_\Omega \Big|E_\e (\frac{u_\e}{\e^\beta}) + \frac12 \nabla_\e \frac{w_\e}{\e^{\beta/2}}\otimes \nabla_\e\frac{w_\e}{\e^{\beta/2}} -  \frac{\mathring{w}'^2}{2\e^\beta} \mathsf{e}_2 \otimes \mathsf{e}_2 \Big|^2\, dx. \end{aligned}
	$$
	For convenience, we set
	\begin{equation}\label{defSe}
		S_\e:=E_\e (\frac{u_\e}{\e^\beta}) + \frac12 \nabla_\e \frac{w_\e}{\e^{\beta/2}}\otimes \nabla_\e\frac{w_\e}{\e^{\beta/2}} -  \frac{\mathring{w}'^2}{2\e^\beta} \mathsf{e}_2 \otimes \mathsf{e}_2.
	\end{equation}

	Our first result establishes partial compactness properties for  all  $0<\beta\le 2$.

	\begin{lemma}\label{compactness4}
		Let $0< \beta \le 2$.
		Let $(u_\e, w_{\e})\subset \calX$ be a sequence that satisfies \eqref{cmp00}.
		Then there exist $r \in H^1(I)$, $\vartheta \in H^1(I)$,   $\gamma_{11}, \gamma_{22}\in L^2(\Omega)$
		such that, up to a subsequence,
		\begin{align}
			&\partial_{2} \frac{w_\e}{\e} - \mathring{w} ' \weak \vartheta &&in \ H^1(\Omega), \label{cmp1}\\
			&\nabla^2_\e w_\e - \frac{\mathring{w} ''}{\e}\mathsf{e}_2\otimes \mathsf{e}_2 \weak
			\begin{pmatrix}
				\gamma_{11} &  \vartheta'
				\\
				\vartheta' & \gamma_{22}
			\end{pmatrix}
			&& in \ L^{2}(\Omega, \Sym), \label{cmp2}\\
			&\frac{w_\e}{\e^{\beta/2}} \weak w=\begin{cases}
				r & \text{if }\beta <2\\
				r+x_2\vartheta+\mathring{w} & \text{if }\beta =2
			\end{cases}  &&in \ H^1(\Omega).\label{cmp3}
		\end{align}
		Also, there exist $S \in L^2(\Omega,\mathbb{R}^{2\times 2}_{\rm sym})$, $(\xi_1, \xi_2) \in BV(I)\times BH(I)$,  such that, up to a subsequence,
		\begin{align}
			S_\e &\weak S&& in \ L^{2}(\Omega, \Sym), \label{cmp5}\\
			\frac{u_{\e}}{\e^\beta} &\weakstar u &&in \ BD(\Omega, \bbR^2),\label{cmp6}
		\end{align}
		with $u$ given by
		\begin{equation}\label{cmp7}
			u = \begin{cases}
				(\xi_1 - x_2 \xi_2', \xi_2) & \text{ if } \beta < 2\\
				(\xi_1 -x_2(\xi_2'+r'\vartheta + c_1\vartheta') - \mathring{w} r'   -\ringring{w}_{\langle 0 \rangle}\vartheta',\\
				\hspace{5cm}		\xi_2 - \frac12 x_2 \vartheta^2 - \mathring{w}\vartheta) & \text{ if } \beta = 2.
			\end{cases}
		\end{equation}
	\end{lemma}
	\begin{proof}
		Let $(u_\e, w_{\e})\subset \calX$ be a sequence that satisfies \eqref{cmp00}.
		Then
		\begin{equation}\label{bound2}
			\sup_{\e} \Big\|\nabla^2_\e w_\e- \frac{\mathring{w} ''}{\e} \mathsf{e}_2 \otimes \mathsf{e}_2\Big\|_{L^2(\Omega,\Sym)} <\infty \quad
			\text{and}\quad\sup_{\e} \n{S_\e}_{L^2(\Omega,\Sym)} <\infty,
		\end{equation}
		with $S_\e$ as defined in \eqref{defSe}.
		From the first of \eqref{bound2}, we deduce that up to a subsequence
		\begin{equation}\label{supd2220}
			\partial_{11} w_\e \weak \gamma_{11}  \quad\text{in } L^2(\Omega), \quad \partial_{12} \frac{w_\e}{\e} \weak \gamma_{12}\quad\text{in } L^2(\Omega),\quad
			\partial_{22} \frac{w_\e}{\e^2} - \frac{\mathring{w} ''}{\e} \weak \gamma_{22}  \quad\text{in } L^2(\Omega),
		\end{equation}
		with $\gamma_{11}, \gamma_{12},\gamma_{22}\in L^2(\Omega).$
		By the second and the third of \eqref{supd2220}, and the fact that $\mathring{w}$ is independent of $x_1$, it follows that
		$\nabla (\partial_2 w_\e/\e-\mathring{w}')$ is a bounded sequence in $L^2(\Omega, \bbR^2)$.
		By Poincar\'e-Wirtinger's inequality,  $\partial_2 w_\e/\e-\mathring{w}'$ is bounded in $H^1(\Omega)$, and hence

		\begin{equation}\label{supd2220011}
			\partial_{2} \frac{w_\e}{\e} - \mathring{w} '\weak \vartheta \quad\text{in }H^1(\Omega),
		\end{equation}
		up to a subsequence, for some $\vartheta\in H^1(\Omega)$.
		From the third of \eqref{supd2220}, we deduce that $\vartheta$ is independent of $x_2$ and therefore
		$\vartheta \in H^1(I)$. Also, from the second of \eqref{supd2220} we have that $\gamma_{12}=\vartheta'$.

		The $11$-component of $S_\e$ is
		$$
		S^\e_{11}=\frac{\partial_1 u^\e_1}{\e^\beta}+\frac 12 (\frac{\partial_1 w_\e}{\e^{\beta/2}})^2.
		$$
		From the second of \eqref{bound2} and  Jensen's inequality, we find
		$$
		\sup_\e \Big(\int_\Omega \frac{\partial_1 u^\e_1}{\e^\beta}+\frac 12 (\frac{\partial_1 w_\e}{\e^{\beta/2}})^2\, dx\Big)^2 < \infty,
		$$
		and by using condition \eqref{condvar} it follows that
		$$ \sup_\e \Big(\Lambda_1 + \int_\Omega (\frac{\partial_1 w_\e}{\e^{\beta/2}})^2 \, dx\Big)^2 < \infty$$
		and
		\begin{equation} \label{supd1222}
			\sup_{\e} \|{\partial_1 \frac{w_\e}{\e^{\beta/2}}}\|_{L^2(\Omega)} <\infty.
		\end{equation}
		This inequality, the fact that  $\beta\le 2$,  \eqref{supd2220011}, and Poincaré's inequality,
		imply that $(\frac{w_\e}{\e^{\beta/2}})$ is a bounded sequence in $H^1(\Omega)$. Hence, up to a subsequence,
		$$ \frac{w_\e}{\e^{\beta/2}} \weak w \qquad \text{ in } H^1(\Omega),
		$$
		for some  $w\in H^1(\Omega)$. By \eqref{supd2220011} it follows that $\partial_2 w=\mathring{w}'+\vartheta$ if $\beta=2$ and
		$\partial_2 w=0$ if $\beta<2$. Thence,
		$$
		w(x_1,x_2)=\begin{cases}
			r(x_1) & \text{if }\beta <2\\
			r(x_1)+x_2\vartheta(x_1)+\mathring{w}(x_2) & \text{if }\beta =2
		\end{cases}
		$$
		for some $r \in H^1(I)$.

		From the second of \eqref{bound2} it immediately follows \eqref{cmp5}, and also that
		$$ \sup_\e\Big\|E \frac{u_\e}{\e^\beta} + \frac12 \nabla \frac{w_\e}{\e^{\beta/2}} \otimes \nabla \frac{w_\e}{\e^{\beta/2}} -\frac12 \e^{2-\beta}\mathring{w}'^2 \mathsf{e}_2 \otimes \mathsf{e}_2\Big\|_{L^2(\Omega)} < +\infty.$$
		This inequality and  \eqref{cmp3} imply  that
		$$ \sup_\e \Big\|E\frac{u_\e}{\e^\beta}\Big\|_{L^1(\Omega,\Sym)} < +\infty.$$
		Since $(u_\e) \subset H^1(\Omega, \bbR^2)$ and \eqref{cond} hold,  we can apply the Korn's inequality \cite[Sec. 2.4]{Temam1987} to deduce that
		$(\frac{u_\e}{\e^\beta})$ is bounded in $BD(\Omega, \bbR^2)$ and that $\frac{u_\e}{\e^\beta} \weakstar u$ in $BD(\Omega, \bbR^2)$ for some $u \in BD(\Omega, \bbR^2)$.

		Since $(S_\e)$ is uniformly bounded in $L^2(\Omega)$, we have
		\begin{equation}\label{Seu}
			\begin{aligned}
				\e S_{12}^\e&=\frac{\partial_1 u_2^\e + \partial_2 u_1^\e}{2\e^\beta}+ \frac12 \frac{\partial_1 w_\e}{\e^{\beta/2}}\frac{\partial_2 w_\e}{\e^{\beta/2}}\to 0 && \text{ in } L^2(\Omega)\\
				\e^2 S_{22}^\e&=\frac{\partial_2 u_2^\e}{\e^\beta} + \frac12 \big(\frac{\partial_2 w_\e}{\e^{\beta/2}}\big)^2- \frac12\e^{2-\beta} \mathring{w}'^2\to 0 && \text{ in } L^2(\Omega).
			\end{aligned}
		\end{equation}
		For $\beta=2$, \eqref{cmp1} implies that $\frac{\partial_2 w_\e}{\e^{\beta/2}}\to \vartheta+\mathring{w}'=\partial_2 w$ in $L^p(\Omega)$ for every $1\le p<+\infty$, while
		for $\beta < 2$ \eqref{cmp1} implies that $\frac{\partial_2 w_\e}{\e^{\beta/2}}\to 0=\partial_2 w$ in $L^p(\Omega)$ for every $1\le p<+\infty$. Thus from \eqref{cmp3} and \eqref{Seu} we deduce that
		$$ \begin{aligned}
			\frac{\partial_1 u_2^\e + \partial_2 u_1^\e}{2\e^\beta} &\to -\frac12\partial_1 w\partial_2 w=\frac{\partial_1 u_2 + \partial_2 u_1}{2}  && \text{ in } L^p(\Omega)\text{ for }1\le p <2,\\
			\frac{\partial_2 u_2^\e}{\e^\beta} &\to - \frac12(\partial_2 w)^2 + \frac12\begin{cases}
				0 & \beta < 2\\
				\mathring{w}'^2 & \beta=2
			\end{cases}=\partial_2 u_2
			&& \text{ in } L^2(\Omega).
		\end{aligned}$$
		By integration and arguing as in Lemma \ref{lemmabd}, we deduce that there are functions $(\xi_1, \xi_2) \in BV(I)\times BH(I)$ such that $(u_1, u_2)$ is as in the statement of the lemma.
	\end{proof}

	We now state a Lemma that we will systematically take advantage of.
	\begin{lemma}\label{lemmaS} Let $0<\beta \le 2$.
		Let $(u_\e, w_{\e})\subset H^1(\Omega,\bbR^2)\times H^2(\Omega)$ and let $S_\e$ as defined in \eqref{defSe}.
		Then for every $\e>0$
		\begin{multline}\label{Seq}
			\int_\Omega -\varphi \det \nabla^2\frac{w_\e}{\e^{\beta/2}}\,dx = \\\int_\Omega S^\e_{11}\partial_{22}\varphi-2\e S^\e_{12}\partial_{12}\varphi+\e^2S^\e_{22}\partial_{11}\varphi\, dx
			 \qquad \forall \ \varphi\in H^2_0(\Omega).
		\end{multline}
		If additionally $S_\e$ is such that $\sup_\e\n{S_\e}_{L^2(\Omega)}<\infty$,  there exists a constant $C>0$ such that for every $\e>0$
		\begin{multline}\label{Sineq}
			\Big|\int_\Omega\varphi \det \nabla^2\frac{w_\e}{\e^{\beta/2}} \,dx\Big|\le \\ C\big(\|\partial_{22}\varphi\|_{L^2(\Omega)}+\e\|\partial_{12}\varphi\|_{L^2(\Omega)}+\e^2\|\partial_{11}\varphi\|_{L^2(\Omega)}\big) \qquad  \forall \ \varphi\in H^2_0(\Omega).
		\end{multline}
		If $\sup_\e\n{\nabla_\e^2 w_\e - \mathring{w}''/\e \mathsf{e}_2\otimes\mathsf{e}_2}_{L^2(\Omega)}<\infty$, there exists a sequence $(d_\e)$ bounded in $L^1(\Omega)$ for which
		\begin{equation}\label{eqdet}
			\det \nabla^2\frac{w_\e}{\e^{\beta/2}}=\e^{2-\beta}d_\e+\e^{1-\beta/2}\partial_{11}\frac{w_\e}{\e^{\beta/2}}\mathring{w}''
		\end{equation}
		for every $\e>0$.
	\end{lemma}
	\begin{proof}
		By \eqref{ccA}, \eqref{cceu}, \eqref{ccww}, and since $\det\nabla^2 \mathring{w} = 0$, we have the identity
		\begin{equation}
			\begin{aligned}
				\int_{\Omega} \varphi \det\nabla^2 w_\e &\, dx = -\int_{\Omega} (E u_\e+ \frac12 \nabla w_\e \otimes \nabla w_\e -\frac12 \e^2 \mathring{w}'^2 \mathsf{e}_2\otimes \mathsf{e}_2) \cdot \cof\nabla^2 \varphi \; dx\\
				&= -\e^2\int_{\Omega} (E_\e u_\e+ \frac12 \nabla_\e w_\e \otimes \nabla_\e w_\e -\frac12 \mathring{w}'^2 \mathsf{e}_2\otimes \mathsf{e}_2) \cdot \cof\nabla^2_\e \varphi \, dx\\
				&= -\e^{2+\beta}\int_{\Omega} \frac{1}{\e^\beta}(E_\e u_\e+ \frac12 \nabla_\e w_\e \otimes \nabla_\e w_\e -\frac12 \mathring{w}'^2 \mathsf{e}_2\otimes \mathsf{e}_2) \cdot \cof\nabla^2_\e \varphi \, dx
			\end{aligned}
		\end{equation}
		which is equivalent to \eqref{Seq} after a little manipulation.

		Inequality \eqref{Sineq} follows from \eqref{Seq} by H\"older'sinequality and the boundedness of $S_\e$ in $L^2(\Omega)$.

		The identity \eqref{eqdet} can be easily verified by posing
		$$d_\e=\det \big(\nabla^2_\e w_\e-\frac{\mathring{w}''}{\e}\mathsf{e}_2 \otimes \mathsf{e}_2\big).$$
		That $d_\e$ is bounded in $L^1(\Omega)$ follows by recalling the second and the first of \eqref{eq:detAB}.
	\end{proof}

	\begin{remark}\label{remarkext}
		Let $\widetilde{\Omega}$ be the extended domain $(-\frac{3}{2}\ell, \frac{3}{2}\ell)\times(-\frac12 ,\frac12 ) =: \widetilde{I} \times W$.
		For every $\e>0$ and for every $(u_\e, w_\e) \in H^1(\Omega, \bbR^2)\times H^2(\Omega)$ we consider the following extension
		\begin{align}
			\widetilde{u}_{1\e}(x) &:= \begin{cases}
				0 & \text{ if } x_1 \in (-\frac{3}{2}\ell, -\frac{\ell}{2})\\
				u_{1\e}(x) & \text{ if } x_1 \in (-\frac{\ell}{2}, \frac{\ell}{2})\\
				\begin{multlined}[t]\e^\beta\Lambda_1 - \e^{1+\beta/2}\Phi_2\Phi_1 x_2 - \e^{1+\beta/2}\Phi_2\mathring{w}\\ - \e^\beta\frac{\Phi_2^2}{2}(x_1-\frac{\ell}{2}) - \e^\beta\Phi_3 x_2 \end{multlined} & \text{ if } x_1 \in (\frac{\ell}{2}, \frac{3}{2}\ell)	\end{cases},\label{ext01}\\
			\widetilde{u}_{2\e}(x) &:= \begin{cases}
				0 & \text{ if } x_1 \in (-\frac{3}{2}\ell, -\frac{\ell}{2})\\
				u_{2\e}(x) & \text{ if } x_1 \in (-\frac{\ell}{2}, \frac{\ell}{2})\\
				\e^\beta \Lambda_2 - \e^2\frac{\Phi_1^2}{2} x_2 - \e^2\Phi_1\mathring{w} + \e^\beta\Phi_3 (x_1-\frac{\ell}{2}) & \text{ if } x_1 \in (\frac{\ell}{2}, \frac{3}{2}\ell) \end{cases},\label{ext02}\\
			\widetilde{w}_{\e}(x) &:= \begin{cases}
				0 & \text{ if } x_1 \in (-\frac{3}{2}\ell, -\frac{\ell}{2})\\
				w_{\e}(x) & \text{ if } x_1 \in (-\frac{\ell}{2}, \frac{\ell}{2})\\
				\e^{\beta/2}\Lambda_3 + \e\mathring{w} + \e x_2 \Phi_1 + \e^{\beta/2} \Phi_2(x_1-\frac{\ell}{2}) & \text{ if } x_1 \in (\frac{\ell}{2}, \frac{3}{2}\ell) \end{cases}.\label{ext03}
		\end{align}
		By \eqref{cond},  $(\widetilde{u}_{\e}, \widetilde{w}_\e) \subset H^1(\widetilde{\Omega},\bbR^2)\times H^2(\widetilde{\Omega})$.
		Moreover, we define
		\begin{equation}\label{defSetilde}
			\widetilde{S}_\e:= E_\e (\frac{\widetilde{u}_\e}{\e^\beta}) + \frac12 \nabla_\e \frac{\widetilde{w}_\e}{\e^{\beta/2}}\otimes \nabla_\e\frac{\widetilde{w}_\e}{\e^{\beta/2}} -  \frac{\mathring{w}'^2}{2\e^\beta} \mathsf{e}_2 \otimes \mathsf{e}_2\end{equation}
		and
		\begin{equation}
			\begin{aligned}
			\widetilde{F}^\beta_\e(\widetilde{u}_\e, \widetilde{w}_\e):=  &\frac12\int_{\widetilde{\Omega}}  \Big| \nabla^2_\e \widetilde{w}_\e- \frac{\mathring{w} ''}{\e} \mathsf{e}_2 \otimes \mathsf{e}_2\Big|^2 \, dx\\ + &\frac12\int_{\widetilde{\Omega}} \Big|E_\e (\frac{\widetilde{u}_\e}{\e^\beta}) + \frac12 \nabla_\e \frac{\widetilde{w}_\e}{\e^{\beta/2}}\otimes \nabla_\e\frac{\widetilde{w}_\e}{\e^{\beta/2}} -  \frac{\mathring{w}'^2}{2\e^\beta} \mathsf{e}_2 \otimes \mathsf{e}_2 \Big|^2\, dx.
			\end{aligned}
		\end{equation}
		The main property of such an extension is that the energy on $\widetilde{\Omega}\setminus\Omega$ is null, in the sense that
		\begin{equation}\label{equiv}
			\hat{F}^\beta_\e(u_\e, w_\e) = \widetilde{F}^\beta_\e(\widetilde{u_\e}, \widetilde{w}_\e) \qquad \forall \e>0.\end{equation}
		Moreover, the analogous counterpart of \eqref{condvar} holds for $\widetilde{u}_{1\e}$.
		Arguing as in the proof of Lemma \ref{compactness4}, we deduce the existence of functions $\widetilde{\xi}_1 \in BV(\widetilde{I})$, $\widetilde{\xi}_2 \in BH(\widetilde{I})$, $\widetilde{r} \in H^1(\widetilde{I})$, $\widetilde{\vartheta} \in H^1(\widetilde{I})$ such that
		\begin{equation}\label{18ext}\frac{\partial_{2}\widetilde{w}_\e}{\e} \weak \widetilde{\vartheta}\qquad in \ H^1(\widetilde{\Omega}),\end{equation}
		\begin{equation}\label{19ext}\frac{\widetilde{w}_\e}{\e^{\beta/2}} \weak \widetilde{w}=\begin{cases}
				\widetilde{r} & \text{if }\beta <2\\
				\widetilde{r}+x_2\widetilde{\vartheta}+\mathring{w} & \text{if }\beta =2
			\end{cases}  \qquad in \ H^1(\widetilde{\Omega}),\end{equation}
		and
		$\frac{\widetilde{u}_{\e}}{\e^\beta} \weakstar \widetilde{u}$ in $BD(\widetilde{\Omega}, \bbR^2)$
		with $\widetilde{u}$ given by
		\begin{equation}\label{uext}
			\widetilde{u}(x_1,x_2) = \begin{cases}
				(\widetilde{\xi}_1 - x_2 \widetilde{\xi}_2', \widetilde{\xi}_2) & \text{ if } \beta < 2\\
				(\widetilde{\xi}_1 -x_2(\widetilde{\xi}_2'+\widetilde{r}'\widetilde{\vartheta} + c_1\widetilde{\vartheta}') - \mathring{w} \widetilde{r}'   -\ringring{w}_{\langle 0 \rangle}\widetilde{\vartheta}',\\
				\hspace{5cm}		\widetilde{\xi}_2 - \frac12 x_2 \widetilde{\vartheta}^2 - \mathring{w}\widetilde{\vartheta}) & \text{ if } \beta = 2
			\end{cases}.
		\end{equation}
		Equivalently, if $(\xi_1, \xi_2, r, \vartheta) \in BV(I)\times BH(I)\times H^1(I) \times H^1(I)$ are as in Lemma \ref{compactness4} it can be easily seen by passing to the limit in \eqref{ext01}--\eqref{ext03} that
		\begin{align}
			\widetilde{\xi}_1 &= \begin{cases}
				0 & \text{ if } x_1 \in (-\frac{3}{2}\ell, -\frac{\ell}{2})\\
				\xi_1 & \text{ if } x_1 \in (-\frac{\ell}{2}, \frac{\ell}{2})\\
				\Lambda_1 - \frac12 \Phi_2^2(x_1-\ell/2)
				& \text{ if } x_1 \in (\frac{\ell}{2}, \frac{3}{2}\ell)
			\end{cases},\label{ext1}\\
			\widetilde{\xi}_2 &= \begin{cases}
				0 & \text{ if } x_1 \in (-\frac{3}{2}\ell, -\frac{\ell}{2})\\
				\xi_2 & \text{ if } x_1 \in (-\frac{\ell}{2}, \frac{\ell}{2})\\
				\Lambda_2 + \Phi_3(x_1-\ell/2) & \text{ if } x_1 \in (\frac{\ell}{2}, \frac{3}{2}\ell)
			\end{cases},\label{ext2}\\
			\widetilde{r} &= \begin{cases}
				0 & \text{ if } x_1 \in (-\frac{3}{2}\ell, -\frac{\ell}{2})\\
				r & \text{ if } x_1 \in (-\frac{\ell}{2}, \frac{\ell}{2})\\
				\Lambda_3 + \Phi_2(x_1-\ell/2) & \text{ if } x_1 \in (\frac{\ell}{2}, \frac{3}{2}\ell)
			\end{cases},\label{ext3}\\
			\widetilde{\vartheta}&= \begin{cases}
				0 & \text{ if } x_1 \in (-\frac{3}{2}\ell, -\frac{\ell}{2})\\
				\vartheta & \text{ if } x_1 \in (-\frac{\ell}{2}, \frac{\ell}{2})\\
				\Phi_1 & \text{ if } x_1 \in (\frac{\ell}{2}, \frac{3}{2}\ell)
			\end{cases}.\label{ext4}\end{align}

	\end{remark}

	\subsection{Finer compactness for $0<\beta < 2$}
	The main efforts are now to show the enhanced regularity (than naively expected) of the limit functions and to characterize the boundary value traces.
	When $0<\beta<2$ it is convenient to work with the extension defined in Remark \ref{remarkext}.

	\begin{lemma}\label{mulambda}
		Let $0<\beta <2$. 		Let $(u_\e, w_{\e})\subset \calX$ and  $(\xi_1, \xi_2, r, \vartheta) \in BV(I)\times BH(I)\times H^1(I) \times H^1(I)$ as in Lemma \ref{compactness4}.
		Then, with the notation of Remark \ref{remarkext},
		\begin{enumerate}[(i)]
			\item $r(-\frac{\ell}{2}) = 0$, $r(\frac{\ell}{2}) = \Lambda_3$, $\vartheta(-\frac{\ell}{2}) = 0$, $\vartheta(\frac{\ell}{2}) = \Phi_1$;
			\item $(\xi_1,r) \in \calB := \{(\xi_1,r) \in BV(I)\times H^1(I) \st \exists \widetilde{\lambda} \in \calM^+(\widetilde{I}), \ \widetilde{\xi}_1 + \widetilde{\lambda} +\frac12 \widetilde{r}'^2\in L^2(\widetilde{I})\}$;
			\item $\widetilde{S}^\e_{11} \weak \widetilde{S}_{11}$ in $L^2(\widetilde{\Omega})$ where 		$
			\widetilde{S}_{11}=(\widetilde{\xi}_1'+ \frac12 \widetilde{r}'^2\calL+\widetilde{\lambda})\otimes \calL  - \widetilde{\xi}_2''\otimes x_2 \calL$
			where $\widetilde{\lambda} \in \calM^+(\widetilde{I})$;
			\item $\xi_2 \in H^2(I)$;
			\item $\xi_2(-\frac{\ell}{2})=0$, $\xi_2(\frac{\ell}{2})=\Lambda_2$; $\xi_2'(-\frac{\ell}{2})=0$, $\xi_2'(\frac{\ell}{2})=\Phi_3$.
		\end{enumerate}
	\end{lemma}

	\begin{proof}
		Let $(u_\e, w_{\e})\subset \calX$ as in Lemma \ref{compactness4}.\\
		Claim \textit{(i)}.	The claim follows from \eqref{ext3}, and \eqref{ext4}, since $\widetilde{r}$ and $\widetilde{\vartheta}$ belong to $H^1(\widetilde{I})$.

		Claims \textit{(ii)} and \textit{(iii)}. By \eqref{equiv} and the uniform bounds \eqref{bound2} and \eqref{supd1222} we deduce that $(\partial_1 \frac{\widetilde{w}_\e}{\e^{\beta/2}})^2$ is uniformly bounded in $L^1(\widetilde{\Omega})$. Thus, up to subsequences, $(\partial_1\frac{\widetilde{w}_\e}{\e^{\beta/2}})^2 $ converges weakly* in the sense of measures to some $\widetilde{\nu} \in \calM^+(\widetilde{\Omega})$. By convexity and \eqref{19ext}, we have that $\widetilde{\nu} \geq (\partial_1 \widetilde{w})^2\calL^2 = \widetilde{r}'^2 \calL \otimes \calL$ since $\beta < 2$.
		For every $\varphi \in C_c^\infty(\widetilde{\Omega})$, by the above convergence it follows that
		$$ \int_{\widetilde{\Omega}} \partial_2\varphi \, d\widetilde{\nu} = \lim_{\e\downarrow 0}\int_{\widetilde{\Omega}} (\frac{\partial_1 \widetilde{w}_\e}{\e^{\beta/2}})^2 \partial_2\varphi \, dx = -2 \lim_{\e\downarrow 0} \int_{\widetilde{\Omega}} \frac{\partial_1 \widetilde{w}_\e}{\e^{\beta/2}}\frac{\partial_{12} \widetilde{w}_\e}\e \frac\e{\e^{\beta/2}} \varphi \, dx = 0, $$
		by \eqref{18ext} and \eqref{19ext}. By Corollary \ref{corLp}, there exists $\widetilde{\rho} \in \calM^+(\widetilde{I})$ such that $\widetilde{\nu} = \widetilde{\rho} \otimes \calL \geq \widetilde{r}'^2\calL \otimes \calL$. Hence, the measure $2\widetilde{\lambda}:=\widetilde{\rho} - \widetilde{r}'^2 \calL \in \calM^+(\widetilde{I})$ and
		\begin{equation}\label{w1sq}
			\frac12(\frac{\partial_1 \widetilde{w}_\e}{\e^{\beta/2}})^2 \weakstar (\frac12\widetilde{r}'^2 \calL + \widetilde{\lambda})\otimes \calL \qquad \text{in} \ \calM(\widetilde{\Omega}).
		\end{equation}
		Let $\widetilde{S}_\e$ as defined in \eqref{defSetilde}. By \eqref{cmp5} and \eqref{equiv} we have that $(\widetilde{S}_\e)_{11} \weak \widetilde{S}_{11}$ in $L^2(\widetilde{\Omega})$.
		From \eqref{uext} and \eqref{w1sq} it follows
		$$
		\begin{aligned}
			\widetilde{S}_{11}=(\widetilde{\xi}_1'+ \frac12 \widetilde{r}'^2\calL+\widetilde{\lambda})\otimes \calL  - \widetilde{\xi}_2''\otimes x_2 \calL.
		\end{aligned}
		$$
		Note that Lemma \ref{lemmaS} holds if we replace $\Omega$ with $\widetilde{\Omega}$, $S_\e$ with $\widetilde{S}_\e$, and $(u_\e, w_\e)$ with $(\widetilde{u}_\e, \widetilde{w}_\e)$.
		By (the counterpart of) \eqref{eqdet}, we find that
		$\det \nabla^2\frac{\widetilde{w}_\e}{\e^{\beta/2}}\to 0$ in $L^1(\widetilde{\Omega})$. By passing to the limit in (the counterpart of) \eqref{Seq} we also deduce that
		$$
		\int_{\widetilde{\Omega}} \widetilde{S}_{11}\partial_{22}\varphi\, dx=0
		$$
		for all $\varphi\in H^2_0(\widetilde{\Omega})$. Hence, there exist two functions $\widetilde{f},\widetilde{g}\in L^2(\widetilde{I})$ such that
		$
		\widetilde{S}_{11}(x_1,x_2)=\widetilde{f}(x_1)-x_2 \widetilde{g}(x_1)
		$.
		By uniqueness of the weak limits, it follows that  $\widetilde{\xi}_1' + \frac12 \widetilde{r}'^2\calL + \widetilde{\lambda} = \widetilde{f}\calL$ and	$\widetilde{\xi}_2'' = \widetilde{g}\calL$. Thus, $\widetilde{\xi}_2'' \in L^2(\widetilde{I})$ and $\widetilde{\xi}_1' + \frac12 \widetilde{r}'^2\calL + \widetilde{\lambda}  \in L^2(\widetilde{I})$.

		Claims \textit{(iv)} and \textit{(v)}. The claims follow from \eqref{ext2} the fact that $\widetilde{\xi}_2 \in H^2(\widetilde{I})$.
	\end{proof}

	\subsection{Finer compactness for $\beta=2$}

	We now focus on the case $\beta=2$. We break the proof into several lemmas. Unlike the case $\beta<2$, it is not convenient to restore immediately to the extension introduced in Remark \ref{remarkext}.

	\begin{lemma}\label{lemmaunif} Let $\beta =2$.
		Let $(u_\e, w_{\e})\subset \calX$  as in Lemma \ref{compactness4}. Then,
		$\frac{w_\e}{\e} \to w$ in $L^\infty(\Omega)$.
	\end{lemma}
	\begin{proof}
		Let $\hat{\frac{w_\e}{\e}}$ denote the extension to $\mathbb{R}^2$ of $\frac{w_\e}{\e}$ for every $\e>0$, as in Lemma \ref{lemmaext}. Let $\hat{\Omega}\supset\supset \Omega$ be the common (compact) support for all the terms in $(\hat{\frac{w_\e}{\e}})$.
		For every $\e>0$, and for every $x=(x_1,x_2),y=(y_1,y_2) \in \Omega$, we have on one hand
		$$ \begin{aligned}
			|\frac{w_\e}{\e}(y) - \frac{w_\e}{\e}(x)| &= |\hat{\frac{w_\e}{\e}}(y) - \hat{\frac{w_\e}{\e}}(x)|\\
			&= |\int_{-\infty}^{y_1}\int_{x_2}^{y_2}\partial_{12}\hat{\frac{w_\e}{\e}}(s,t) \, dt\,ds + \int_{-\infty}^{x_2}\int_{x_1}^{y_1}\partial_{12}\hat{\frac{w_\e}{\e}}(s,t) \, ds\,dt|\\
			&= |\int_{-\infty}^{y_1}\int_{x_2}^{y_2}\mathbf{1}_{\hat{\Omega}}\partial_{12}\hat{\frac{w_\e}{\e}}(s,t) \, dt\,ds + \int_{-\infty}^{x_2}\int_{x_1}^{y_1}\mathbf{1}_{\hat{\Omega}}\partial_{12}\hat{\frac{w_\e}{\e}}(s,t) \, ds\,dt|\\
			&\leq C (\sqrt{|y_2-x_2|} + \sqrt{|y_1-x_1|})\n{\partial_{12}\hat{\frac{w_\e}{\e}}}_{L^2(\bbR^2)}\\
			&\leq C (\sqrt{|y_2-x_2|} + \sqrt{|y_1-x_1|})(\n{\partial_{12}\frac{w_\e}{\e}}_{L^2(\Omega)} + \n{\frac{w_\e}{\e}}_{H^1(\Omega)})\\
			&\leq C (\sqrt{|y_2-x_2|} + \sqrt{|y_1-x_1|})
		\end{aligned}$$
		where we used H\"older's inequality, the boundedness of $\hat{\Omega}$, Lemma \ref{lemmaext}, and \eqref{cmp2}, \eqref{cmp3}. Thus, $(\frac{w_\e}{\e})$ is equicontinuous on $\bar\Omega$, the closure of $\Omega$.

		On the other hand, for every $\e>0$ and every $x=(x_1,x_2) \in \bar\Omega$, we have
		$$ \begin{aligned}
			|\frac{w_\e}{\e}(x)|  = |\hat{\frac{w_\e}{\e}}(x)| &= |\int_{-\infty}^{x_1}\int_{-\infty}^{x_2}\partial_{12} \hat{\frac{w_\e}{\e}}(s,t) \, dt\,ds|\\
			&\leq \n{\partial_{12}\hat{\frac{w_\e}{\e}}}_{L^1(\bbR^2)}
			= \n{\partial_{12}\hat{\frac{w_\e}{\e}}}_{L^1(\hat{\Omega})}\\
			&\leq C (\n{\partial_{12}\frac{w_\e}{\e}}_{L^2(\Omega)} + \n{\frac{w_\e}{\e}}_{H^1(\Omega)}) \leq C.
		\end{aligned}$$
		By Ascoli-Arzelà's theorem, recalling that $\frac{w_\e}{\e} \weak w$ in $H^1(\Omega)$ by Lemma \ref{compactness4}, we deduce that up to a subsequence $\frac{w_\e}{\e} \to w$ uniformly. By Urysohn's lemma, the full sequence $(\frac{w_\e}{\e})$ converges.
	\end{proof}

	\begin{lemma}\label{lemmamu2}
		Let $\beta=2$. Let $(u_\e, w_\e)\subset \calX$ as in Lemma \ref{compactness4}. Then, there exists $\mu \in \calM^+(\Omega)$ such that
		\begin{equation}\label{cmp4}(\partial_1 \frac{w_\e}{\e^{\beta/2}})^2 \calL^2 \weakstar (\partial_1 w)^2\calL^2 + 2\mu \qquad  in \ \calM(\Omega).\end{equation}
	\end{lemma}
	\begin{proof}
		By the uniform bounds \eqref{bound2} and \eqref{supd1222} we deduce that $(\partial_1 \frac{w_\e}{\e^{\beta/2}})^2$ is uniformly bounded in $L^1({\Omega})$. Thus, up to subsequences, $(\partial_1\frac{{w}_\e}{\e^{\beta/2}})^2 $ converges weakly* in the sense of measures to some ${\nu} \in \calM({\Omega})$. By convexity, we have that ${\nu} \geq (\partial_1 {w})^2\calL^2$. We set $2\mu:={\nu} - (\partial_1 {w})^2\calL^2$ which clearly belongs to $\calM^+({\Omega})$. 
	\end{proof}

	The next result will call upon the finite set $Y$ defined in \eqref{defY}.

	\begin{lemma}\label{lemmawsq}
		Let $\beta =2$.
		Let $\mu$ as in Lemma \ref{lemmamu2}. Then, the support of $\mu$ is contained in $I\times Y$.
		In particular,
		for  every $y\in Y$ there is $\lambda_y \in \calM^+(I\times \{y\})$ such that
		$$
		\mu = \sum_{y\in Y} \lambda_y\otimes \delta_y.
		$$
	\end{lemma}
	\begin{proof}
		Let $w=r+x_2\vartheta+\mathring{w}$ be as in \eqref{cmp3}.
		As $r,\vartheta \in H^1(I)$, by Lemma \ref{lemmader2} there are  sequences $(r_\e), (\vartheta_\e) \subset C^\infty(\overline{I})$ such that $(r_\e, \vartheta_\e) \to (r, \vartheta)$ in $H^1(I)\times H^1(I)$  and $(\e r_\e'',\e \vartheta_\e'')\to (0,0)$ in $L^2(I)\times L^2(I)$.  The convergence $(r_\e, \vartheta_\e) \to (r, \vartheta)$ is actually uniform by the compact embedding $H^1(I) \embed C^0(\overline{I})$. Thus, the sequence $\tilde w_\e:=r_\e + x_2 \vartheta_\e+\mathring{w}$ converges uniformly to $w$ in $\Omega$.
		Combining this with Lemma \ref{lemmaunif}, we have that
		$z_\e\coloneqq   \frac{w_\e}{\e} - \tilde w_\e$ converges uniformly to $0$ in $\Omega$.
		By using \eqref{cmp3} and \eqref{cmp1} we deduce that
		\begin{equation}\label{ze1}
			\begin{aligned}
				\e\|\partial_1 z_\e\|_{L^2(\Omega)}&=\e\|\partial_1 \frac{w_\e}\e-r_\e'-x_2\vartheta_\e'\|_{L^2(\Omega)}\to 0,\\
				\|\partial_2 z_\e\|_{L^2(\Omega)}&=\|\partial_2 \frac{w_\e}\e-\vartheta_\e-\mathring{w}'\|_{L^2(\Omega)}\to 0
			\end{aligned}
		\end{equation}
		as $\e$ goes to zero. Also, by using \eqref{cmp2} we find that
		\begin{equation}\label{ze2}
			\begin{aligned}
				\e^2\|\partial_{11} z_\e\|_{L^2(\Omega)}&=\e\|\partial_{11} w_\e-\e r_\e'-x_2\e\vartheta_\e'\|_{L^2(\Omega)}\to 0,\\
				\e\|\partial_{12} z_\e\|_{L^2(\Omega)}&=\e\|\partial_{12} \frac{w_\e}\e-\vartheta_\e'\|_{L^2(\Omega)}\to 0,\\
				\|\partial_{22} z_\e\|_{L^2(\Omega)}&=\|\partial_{22} \frac{w_\e}\e-\mathring{w}''\|_{L^2(\Omega)}\to 0.
			\end{aligned}
		\end{equation}
		Let $\psi \in C_c^\infty(\Omega)$. By taking
		$\varphi_\e = \psi z_\e$ in place of $\varphi$ in \eqref{Sineq} we have
		\begin{equation}\label{Sineq2}
			\Big|\int_\Omega\varphi_\e \det \nabla^2\frac{w_\e}{\e} \,dx\Big|\le C\big(\|\partial_{22}\varphi _\e\|_{L^2(\Omega)}+\e\|\partial_{12}\varphi_\e\|_{L^2(\Omega)}+\e^2\|\partial_{11}\varphi_\e\|_{L^2(\Omega)}\big).
		\end{equation}
		From \eqref{ze1} and \eqref{ze2} it immediately follows that the left side of \eqref{Sineq2} goes to zero as $\e\to 0$.

		We now study the right side of \eqref{Sineq2}. By \eqref{eqdet}
		we have
		\begin{align*}
			0=\lim_{\e\to 0}\int_\Omega\varphi_\e \det \nabla^2\frac{w_\e}{\e}\,dx&=\lim_{\e\to 0}\int_\Omega\psi z_\e \big(d_\e+  \mathring{w}''\partial_{11}\frac{w_\e}{\e}\big)\,dx\\&=\lim_{\e\to 0}\int_\Omega\psi z_\e   \mathring{w}''\partial_{11}\frac{w_\e}{\e}\,dx,
		\end{align*}
		where the last equality follows since $d_\e$ is bounded in $L^1(\Omega)$ and $z_\e$ converges uniformly to zero.
		But by \eqref{cmp3} we deduce that
		\begin{multline}
			0=-\lim_{\e\to 0}\int_\Omega\psi z_\e   \mathring{w}''\partial_{11}\frac{w_\e}{\e}\,dx=\lim_{\e\to 0}\int_\Omega\big(z_\e\partial_{1}\psi    +\psi\partial_1z_\e\big)\mathring{w}''\partial_{1}\frac{w_\e}{\e}\,dx \\
			=
			\lim_{\e\to 0}\int_\Omega\psi\mathring{w}''\partial_1z_\e\partial_{1}\frac{w_\e}{\e}\,dx,
		\end{multline}
		that is equivalent to
		$$
		\lim_{\e\to 0}\int_\Omega\psi\mathring{w}''(\partial_1\frac{w_\e}{\e})^2\,dx
		=
		\lim_{\e\to 0}\int_\Omega\psi\mathring{w}''\partial_1\tilde w_\e\partial_{1}\frac{w_\e}{\e}\,dx=\int_\Omega\psi\mathring{w}''(\partial_{1}w)^2\,dx,
		$$
		as $\tilde w_\e\to w$ in $H^1(\Omega)$ and $w_\e/\e\weak w$ in $H^1(\Omega)$.
		So $\mathring{w}''(\partial_1\frac{w_\e}{\e})^2 \to \mathring{w}''(\partial_1 w)^2$ in $\calD'(\Omega)$.
		Since the sequence $(\mathring{w}''(\partial_1\frac{w_\e}{\e})^2)$ is bounded in $L^1(\Omega)$, we have that
		$$\mathring{w}''(\partial_1 \frac{w_\e}{\e})^2 \calL^2 \weakstar \mathring{w}''(\partial_1 w)^2\calL^2  \quad \mbox{in } \calM(\Omega).
		$$
		On the other hand, by \eqref{cmp4}, we also have that
		$$\mathring{w}''(\partial_1 \frac{w_\e}{\e})^2 \calL^2 \weakstar \mathring{w}''(\partial_1 w)^2\calL^2 + 2\mathring{w}''\mu \quad \mbox{in } \calM(\Omega).
		$$
		Hence, by uniqueness of the weak limits, we deduce that $\mathring{w}''\mu=0$.
		This implies that also the total variation of the measure $\mathring{w}''\mu$ is null, that is
		$|\mathring{w}''\mu|=|\mathring{w}''|\mu=0$. We now show that $\mu$ concentrates on $I\times Y$.
		Let $A\subset \Omega\setminus (I\times Y)$ be a Borel set, and let
		$$A_n=\{x\in A:|\mathring{w}''|(x)\ge \frac 1n\}.
		$$
		For $n$ large enough, $A_n$ is nonempty and $(A_n)\subset A$ monotonically converges to $A$. Then $0=
		(|\mathring{w}''|\mu)(A_n)\ge \frac1n\mu(A_n)\ge 0$, that implies  $\mu(A_n)=0$. In turn, we deduce $\mu(A)=0$.  Thus $\mu=\mu\mres (I\times Y)=\sum_{y\in Y} \mu\mres (I\times\{y\})$, since $Y$ is finite.
		For fixed $y\in Y$, let $\lambda_y$ be the Radon measure defined on the Borel sets of $I\times\{y\}$ by
		$$\lambda_y(B_y)=\mu\mres (I\times\{y\})(B_y)=\mu(B_y),
		$$
		for all Borel sets $B_y\subset I\times\{y\}$. For a Borel set $B\subset \Omega$, we have that
		$$\lambda_y\otimes\delta_y(B)=\lambda_y(B\cap \{x_2=y\})=\lambda_y(B\cap( I\times\{y\}))=\mu(B\cap( I\times\{y\}))=\mu\mres (I\times\{y\})(B)
		$$
		and hence $\mu\mres (I\times\{y\})=\lambda_y\otimes\delta_y.$
	\end{proof}

	The characterization of the support of $\mu$ for $\beta=2$ is enough to assess an enhanced regularity for $r$ and $\vartheta$.
	\begin{lemma}\label{rtH2}
		Let $\beta=2$. Let $r$ and $\vartheta$ as in Lemma \ref{compactness4} and $\mu$ as in Lemma \ref{lemmamu2}. Then, $r, \vartheta \in H^2(I)$.
		Moreover, for every $y\in Y$  the measures $\lambda_y \in \calM^+(I\times \{y\})$ defined in Lemma \ref{lemmawsq} are equal to $\lambda_y=h_y\calL$ for some $h_y\in L^2(I)$ and $h_y\ge 0$ almost everywhere in $I$.
		In particular,
		\begin{equation}\label{muh}
			\mu = \sum_{y\in Y} h_y\calL\otimes \delta_y.
		\end{equation}
	\end{lemma}

	\begin{proof}
		Let $(u_\e, w_{\e})\subset \calX$ as in Lemma \ref{compactness4}.
		We start by observing that \eqref{cmp1},  \eqref{cmp3}, and  \eqref{cmp4}
		imply that
		$$
		\nabla \frac{w_\e}{\e}\otimes \nabla \frac{w_\e}{\e}=
		\begin{pmatrix}
			(\frac{\partial_1 w_\e}{\e})^2 & \frac{\partial_1 w_\e}{\e}\frac{\partial_2 w_\e}{\e}\\
			\frac{\partial_1 w_\e}{\e}\frac{\partial_2 w_\e}{\e} & (\frac{\partial_2 w_\e}{\e})^2
		\end{pmatrix}
		\weakstar
		\begin{pmatrix}
			(\partial_1 w)^2\calL^2+2\mu & \partial_1 w\partial_2 w\\
			\partial_1 w\partial_2 w & (\partial_2 w)^2
		\end{pmatrix}
		$$
		in $\calM(\Omega, \mathbb{R}^{2\times 2}_{\rm sym})$.
		Hence, from the identity
		$$
		\int_\Omega \det \nabla^2 \frac{w_\e}{\e} \varphi \, dx =-\frac12\int_\Omega  \nabla \frac{w_\e}{\e}\otimes \nabla \frac{w_\e}{\e} \cdot \cof\nabla^2\varphi dx \qquad \forall  \ \varphi \in C_c^\infty(\Omega)
		$$
		we find that
		$$
		\begin{aligned}
		\lim_{\e\to 0}		\int_\Omega \det \nabla^2 \frac{w_\e}{\e} \varphi \, dx =-\frac12&\int_\Omega (\partial_1 w)^2\partial_{22}\varphi-2 \partial_1 w\partial_2 w\partial_{12}\varphi+(\partial_2 w)^2 \partial_{11}\varphi\, dx\\-&\int_\Omega \partial_{22}\varphi\,d\mu,
				\end{aligned}
		$$
		for every $\varphi \in C_c^\infty(\Omega)$. By means of \eqref{cmp3} we obtain  that
		$$
		\begin{aligned}
			\lim_{\e\to 0}\int_\Omega \det \nabla^2 \frac{w_\e}{\e} \varphi \, dx
			&=-\frac12\int_\Omega (r'+x_2\vartheta')^2\partial_{22}\varphi-2 (r'+x_2\vartheta')(\vartheta+\mathring{w}')\partial_{12}\varphi\\
			&\hspace{4cm}+(\vartheta+\mathring{w}')^2 \partial_{11}\varphi\, dx-\int_\Omega \partial_{22}\varphi\,d\mu\\
			&=-\frac12\int_\Omega 2\vartheta'^2\varphi+2 \big(\vartheta'(\vartheta+\mathring{w}')+(r'+x_2\vartheta')\mathring{w}''\big)\partial_{1}\varphi\\
			&\hspace{4cm}-2\vartheta'(\vartheta+\mathring{w}') \partial_{1}\varphi\, dx-\int_\Omega \partial_{22}\varphi\,d\mu,\\
			&=-\int_\Omega \vartheta'^2\varphi+ (r'+x_2\vartheta')\mathring{w}''\partial_{1}\varphi\, dx-\int_\Omega \partial_{22}\varphi\,d\mu
		\end{aligned}
		$$
		for every $\varphi \in C_c^\infty(\Omega)$.
		By using this convergence result, \eqref{cmp5}, and \eqref{Seq}, we deduce that
		\begin{equation*}
			\int_\Omega S_{11}\partial_{22}\varphi\, dx=
			\int_\Omega \vartheta'^2\varphi+ (r'+x_2\vartheta')\mathring{w}''\partial_{1}\varphi\, dx+\int_\Omega \partial_{22}\varphi\,d\mu
		\end{equation*}
		for every $\varphi \in C_c^\infty(\Omega)$.
		From this equality it follows, also using Lemma \ref{lemmawsq}, that
		\begin{equation}\label{Seqdddd}
			\int_\Omega S_{11}\varphi_1\varphi_2''\, dx=
			\int_\Omega \vartheta'^2\varphi_1\varphi_2+ (r'+x_2\vartheta')\mathring{w}''\varphi_1'\varphi_2\, dx+ \sum_{y\in Y} \varphi_2''(y)\int_{I\times\{y\}} \varphi_1\,d\lambda_y
		\end{equation}
		for every $\varphi_1 \in C_c^\infty(I)$ and $\varphi_2 \in C_c^\infty(W)$.

		By choosing $\varphi_2$ odd and $\varphi_2''=0$ on $Y$, recalling that $\mathring{w}$ is even,
		from \eqref{Seqdddd} we find that
		$$
		\int_\Omega S_{11}\varphi_1\varphi_2''\, dx=
		\int_\Omega x_2\vartheta'\mathring{w}''\varphi_1'\varphi_2\, dx=\int_W x_2\mathring{w}''\varphi_2\, dx_2\int_I \vartheta'\varphi_1'\, dx_1
		$$
		from which we deduce that $|\int_I \vartheta' \varphi_1' \, dx_1| \leq C \n{\varphi_1\varphi_2''}_{L^2(\Omega)}\le C\n{\varphi_1}_{L^2(I)}$, for every $\varphi_1 \in C_c^\infty(I)$.  By Riesz's representation theorem (see also \cite[Proposition 8.3]{Brezis2011}), we deduce $\vartheta \in H^2(I)$.

		Instead, by choosing $\varphi_2$ even and $\varphi_2''=0$ on $Y$ from \eqref{Seqdddd} we find
		$$
		\int_\Omega S_{11}\varphi_1\varphi_2''\, dx=
		\int_\Omega \vartheta'^2\varphi_1\varphi_2+ r'\varphi_1'\varphi_2\, dx,
		$$
		that leads to
		$|\int_I r' \varphi_1' \, dx_1| \leq C(1 + \n{\vartheta'}_{L^4(I)}^{2})\n{\varphi_1}_{L^2(I)}$
		for all  $\varphi_1 \in C_c^\infty(I)$. Thus, $r \in H^2(I)$.

		Finally, let us fix $y\in Y$ and let us also fix $\varphi_2 \in C_c^\infty(W)$ such that
		$\varphi_2''(y)\ne 0$ and $\varphi_2''=0$ on $Y\setminus\{y\}$. From \eqref{Seqdddd} we deduce that
		$$
		\varphi_2''(y)\int_{I\times\{y\}} \varphi_1\,d\lambda_y=
		\int_\Omega S_{11}\varphi_1\varphi_2''\, dx-
		\int_\Omega \big( \vartheta'^2- (r''+x_2\vartheta'')\mathring{w}''\big)\varphi_1\varphi_2\, dx
		$$
		which implies that
		\begin{equation}\label{mubound}
			\Big|\int_{I} \varphi_1\,d\lambda_y\Big|\le C \|\varphi_1\|_{L^2(I)}
		\end{equation}
		for all  $\varphi_1 \in C_c^\infty(I)$. By density, the linear functional $\lambda_y:C_c^\infty(I) \to \mathbb{R}$ defined by $\lambda_y(\varphi_1)=\int_{I} \varphi_1\,d\lambda_y$ can be extended to a linear functional, that we do not rename, $\lambda_y:L^2(I) \to \mathbb{R}$ that satisfies \eqref{mubound} for all $\varphi_1\in L^2(I)$.
		By Riesz's representation theorem there exists $h_y\in L^2(I)$ such that $\int_{I} \varphi_1\,d\lambda_y=\int_{I} \varphi_1h_y\,dx_1$ for all  $\varphi_1 \in C_c^\infty(I)$. Thus, $\lambda_y=h_y\calL$. Since $\lambda_y$ is a positive measure it follows that $h_y\ge 0$ almost everywhere in $I$.
	\end{proof}

	Finally, we can prove the following characterization for $\beta=2$.
	\begin{lemma}\label{b2S11}
		Let $\beta = 2$. Let $(\xi_1, \xi_2) \in BV(I)\times BH(I)$ as in Lemma \ref{compactness4} and   $\mu\in \calM^+(\Omega)$ as in Lemma \ref{lemmamu2}.    Then:
		\begin{enumerate}[(i)]
			\item $(\xi_1, \xi_2) \in H^1(I)\times H^2(I)$;
			\item  $\mu = 0$;
			\item the component 11 of $S$, defined in \eqref{cmp5}, is
			$
			S_{11}=\partial_1 u_1+\frac 12 (\partial_1w)^2
			$.
		\end{enumerate}
	\end{lemma}
	\begin{proof}
		From \eqref{cmp4}, \eqref{cmp5}, and \eqref{cmp6}, we have that
		\begin{equation}\label{S11uwmu}
			S_{11}=\partial_1 u_1+\frac 12 (\partial_1w)^2 +\mu,
		\end{equation}
		where $S_{11} \in L^2(\Omega)$, $\partial_1 u_1\in \calM(\Omega)$, $\mu\in \calM^+(\Omega)$, and $w\in H^1(\Omega).$
		By \eqref{cmp7} we know that
		$$
		u_1=\xi_1 -x_2(\xi_2'+r'\vartheta + c_1\vartheta') - \mathring{w} r'   -\ringring{w}_{\langle 0 \rangle}\vartheta' - x_2\alpha + t_1,
		$$
		and since Lemma \ref{rtH2} states that $r,\vartheta\in H^2(I)$, we can write $\partial_1 u_1$ as
		$$
		\partial_1 u_1=\xi_1' -x_2\xi_2''+f,
		$$
		where $\xi_1',\xi_2''\in\calM(I)$ and $f\in L^2(\Omega)$.
		Also, from \eqref{cmp3} we have that $\partial_1 w=r'+x_2\vartheta'$, and  from Lemma \ref{rtH2} we deduce that $(\partial_1w)^2\in L^2(\Omega)$.
		With these considerations, we can rewrite \eqref{S11uwmu} as
		\begin{equation}\label{S11uwmu2}
			\xi_1' -x_2\xi_2'' +\sum_{y\in Y} h_y\calL\otimes \delta_y=S_{11}-f-\frac 12 (\partial_1w)^2,
		\end{equation}
		where we have used \eqref{muh}. We observe that the right side of \eqref{S11uwmu2} is in $L^2(\Omega)$.
		From this equality, for every
		$\varphi_1 \in C_c^\infty(I)$ and $\varphi_2 \in C_c^\infty(W)$, we infer that
		$$
		\Big|\int_W\varphi_2\,dx_2\int_I\varphi_1\,d\xi_1'-\int_Wx_2\varphi_2\,dx_2\int_I\varphi_1\,d\xi_2''
		\Big|\le C\|\varphi_1\|_{L^2(I)},
		$$
		where the constant $C$ depends on $\|\varphi_2\|_{L^\infty(W)}$ and $ \|h_y\|_{L^2(I)}$.
		Taking $\varphi_2$ even and arguing as at the end of the proof of Lemma \ref{rtH2} we deduce that
		$\xi_1\in H^1(I)$. Similarly, with $\varphi_2$ odd we conclude that $\xi_2\in H^2(I)$.

		With these discoveries and \eqref{S11uwmu2}, we deduce that $\sum_{y\in Y} h_y\calL\otimes \delta_y\in L^2(\Omega)$, which implies $h_y=0$ for every $y\in Y$. Indeed, it suffices to multiply
		\eqref{S11uwmu2} by $\mathbf{1}_{I\times \{z\}}$ with $z\in Y$ and integrate over $\Omega$, to deduce that $h_z=0$ almost everywhere in $I$.
	\end{proof}

	Under the convergences established in Lemma \ref{compactness4} the trace operators are not continuous - apart for $(\frac{w_\e}{\e^{\beta/2}}(\pm \frac{\ell}{2}, x_2))$ - and we cannot readily deduce the boundary values of $r', \vartheta', \xi_1, \xi_2, \xi_2'$. However, exploiting the extension introduced in Remark \ref{remarkext} and the enhanced regularity of the limit functions, we can prove the following result.
	\begin{lemma}
		Let $\beta=2$. Let $r,\vartheta \in H^2(I)$ and $(\xi_1, \xi_2) \in H^1(I)\times H^2(I)$ as in Lemma \ref{compactness4}, \ref{rtH2}, \ref{b2S11}. Then
		\begin{enumerate}[(i)]
			\item $r(-\frac{\ell}{2}) = 0$, $r(\frac{\ell}{2}) = \Lambda_3$, $\vartheta(-\frac{\ell}{2}) = 0$, $\vartheta(\frac{\ell}{2}) = \Phi_1$;
			\item $r'(-\frac{\ell}{2})=0$, $r'(\frac{\ell}{2}) = \Phi_2$;
			\item $\vartheta'(-\frac{\ell}{2}) = \vartheta'(\frac{\ell}{2})=0$;
			\item $\xi_1(-\frac{\ell}{2})=0$, $\xi_1(\frac{\ell}{2}) = \Lambda_1$;
			\item $\xi_2(-\frac{\ell}{2})=0$, $\xi_2(\frac{\ell}{2})=\Lambda_2$; $\xi_2'(-\frac{\ell}{2})=0$, $\xi_2'(\frac{\ell}{2})=\Phi_3$.
		\end{enumerate}
	\end{lemma}
	\begin{proof}
		Claim \textit{(i)}. It is proved directly by using the continuity of the trace under weak convergence in $H^1$.

		Claims \textit{(ii)} - \textit{(v)}.
		Consider the extensions $(\widetilde{u}_\e, \widetilde{w}_\e)$ as defined in Section \ref{sec:comp}.
		We can thus applying the results of Lemmas \ref{compactness4}, \ref{rtH2}, \ref{b2S11}, to deduce the existence of limit functions $(\widetilde{\xi_1}, \widetilde{\xi_2}, \widetilde{r}, \widetilde{\vartheta}) \in H^1(\widetilde{I})\times H^2(\widetilde{I}) \times H^2(\widetilde{I}) \times H^2(\widetilde{I})$ whose expressions are given in \eqref{ext1}--\eqref{ext4}.
		Henceforth, by the aforementioned regularity of $(\widetilde{\xi_1}, \widetilde{\xi_2}, \widetilde{r}, \widetilde{\vartheta})$ and the embeddings $H^1(\widetilde{I}) \embed C^0(\widetilde{I})$, $H^2(\widetilde{I}) \embed C^1(\widetilde{I})$, we deduce the stated boundary trace values on $\partial I$.
	\end{proof}

	\section{The $\Gamma$-limit for $\beta=2$}\label{sec:gamma2}

	We begin with the case $\beta = 2$, which is simpler.

	Recall that $\mathcal{X} = L^1(\Omega; \mathbb{R}^2) \times L^2(\Omega)$. Given $(\Lambda_1, \Lambda_2, \Lambda_3, \Phi_1, \Phi_2, \Phi_3) \in \bbR^6$, define the set of admissible displacements by
	$$
	\begin{aligned}
		\calA^{2} \coloneqq
		\{(u,w) \in \calX \st & \exists   (\xi_1, \xi_2, r, \vartheta) \in H^1(I) \times H^2(I)\times H^2(I)\times H^2(I), \\
		&w = r + x_2\vartheta + \mathring{w}, \\
		&u_{1} =\xi_1 -x_2(\xi_2'+r'\vartheta + c_1\vartheta') - \mathring{w} r'   -\ringring{w}_{\langle 0 \rangle}\vartheta'\\
		&u_{2}=\xi_2 - \frac12 x_2 \vartheta^2 - \mathring{w}\vartheta\\
		& r(-\frac{\ell}{2})=r'(-\frac{\ell}{2})=0, \ r(\frac{\ell}{2})=\Lambda_3, r'(\frac{\ell}{2})=\Phi_2\\
		& \vartheta(-\frac{\ell}{2})=\vartheta'(-\frac{\ell}{2})=0, \ \vartheta(\frac{\ell}{2})=\Phi_1, \vartheta'(\frac{\ell}{2})=0, \\
		& \xi_1(-\frac{\ell}{2})=0, \ \xi_1(\frac{\ell}{2})=\Lambda_1,\\
		&\xi_2(-\frac{\ell}{2})=\xi_2'(-\frac{\ell}{2})=0, \ \xi_2(\frac{\ell}{2})=\Lambda_2, \xi_2'(\frac{\ell}{2})=\Phi_3\}.
	\end{aligned}
	$$
	We observe that a pair $(u, w) \in \mathcal{A}^{2}$ uniquely determines the functions $(\xi_1, \xi_2, r, \vartheta)$ appearing in the definition of $\mathcal{A}^{2}$. Hence, the following definition of a functional, with domain $\mathcal{X}$, is well-defined:
	$$
	F^{2}(u,w) \coloneqq
	\begin{cases}
		\begin{aligned}
			\int_I |\vartheta'|^2 &+ \frac12 |\xi_1'+ \frac12 r'^2 + \frac{1}{24}\vartheta'^2|^2
			\\ & +\frac{J_1}{2} |\xi_2'' + r''\vartheta + c_1\vartheta''|^2  +\frac{J_2}{2} |c_2 r'' - \vartheta'^2|^2\\
			&  +\frac{J_3}{2} |r''|^2  +\frac{J_4}{2} |\vartheta''|^2 \, dx_1 \end{aligned}
			& \text{ if }  (u,w) \in \calA^{2}\\
			+\infty & \text{ otherwise}.
	\end{cases}
	$$

	We now state and prove the $\Gamma$-convergence result for $\beta = 2$. To keep the statement compact, we express the convergence using strong convergence in $\mathcal{X}$, even though, as shown in Section~\ref{sec:comp}, this can be significantly improved.

	\begin{theorem}\label{thm:gc3}
		As $\e\downarrow 0$, the sequence of functionals $(F_{\e}^{2})$ $\Gamma$-converges to $F^{2}$ in $\calX$. Precisely
		\begin{enumerate}[(a)]
			\item (Liminf inequality) for every $(u,w)\in \calX$ and for every sequence $(u_\e, w_\e)\subset \calX$ such that
			$$
			(\frac{u_\e}{\e^2},\frac{w_\e}\e)\to (u,w)\quad \mbox{in }\calX
			$$
			we have
			\begin{equation*}
				\liminf\limits_{\e\downarrow 0} F_{\e}^{2}(u_\e, w_\e)\geq F^{2}(u,w);
			\end{equation*}
			\item (Recovery sequence)
			for every  $(u,w)\in \calX$ there exists a sequence $(u_\e, w_\e)\subset \calX$ such that
			$$
			(\frac{u_\e}{\e^2},\frac{w_\e}\e)\to (u,w)\quad \mbox{in }\calX
			$$
			and
			\begin{equation*}
				\lim\limits_{\e\downarrow 0} F_{\e}^{2}(u_\e, w_\e) = F^{2}(u,w).
			\end{equation*}
		\end{enumerate}
	\end{theorem}
	\begin{proof}
		{\it {(a) (Liminf inequality)}}
		Without loss of generality, let us assume $\liminf_{\e\downarrow 0} F_{\e}^{2}(u_\e, w_\e) <\infty$, otherwise there is  nothing to prove.
		Hence, up to a subsequence, $\sup_{\e}F_{\e}^{2}(u_\e, w_\e) <\infty$. Therefore,  we can rely on the results stated in Section \ref{sec:comp} and in particular deduce that $(u,w) \in \calA^{2}$.

		With the definition \eqref{defSe} and by weak sequential lower semicontinuity we have that
		$$
		\begin{aligned}
			\liminf_{\e\downarrow 0} F_{\e}^{2}(u_\e,w_\e)&\ge
			\liminf_{\e\downarrow 0}\frac{1}{2}\int_\Omega |\nabla^2_\e w_\e - \mathring{w}''\mathsf{e}_2\otimes\mathsf{e}_2|^2+ |S_\e|^2\, dx  \\
			&\ge\frac{1}{2}\int_\Omega |\begin{pmatrix}
				\gamma_{11} & \vartheta'\\ \vartheta' & \gamma_{22}
			\end{pmatrix}|^2\, dx + \frac{1}{2}\int_\Omega |\begin{pmatrix}
				\partial_1 u_1 + \frac12 (\partial_1 w)^2 & S_{12}\\ S_{12} & S_{22}
			\end{pmatrix}|^2\, dx,
		\end{aligned}
		$$
		where we used \eqref{cmp2}, \eqref{cmp5}, and Lemma \ref{b2S11}.
		Hence, since $(u,w) \in \calA^{2}$ we have
		$$
		\begin{aligned}
			\liminf_{\e\downarrow 0} F_{\e}^{2}(u_\e,w_\e)&\ge
			\int_I |\vartheta'|^2 dx_1 + \frac12 \int_\Omega |\partial_1 u_1 + \frac12 (\partial_1 w)^2|^2\,dx\\
			&=\begin{multlined}[t] \int_I |\vartheta'|^2 dx_1 + \frac12 \int_\Omega |\xi_1' - x_2 (\xi_2'' + r''\vartheta + r'\vartheta' + c_1\vartheta'') - \mathring{w} r''\\ - \ringring{w}_{\langle 0 \rangle} \vartheta'' + \frac12 r'^2 +x_2r'\vartheta'+\frac12 x_2^2 \vartheta'^2|^2 \, dx \end{multlined}\\
			&=\begin{multlined}[t] \int_I |\vartheta'|^2 dx_1 + \frac12 \int_\Omega |(\xi_1'+ \frac12 r'^2 + \frac{1}{24}\vartheta'^2) - x_2 (\xi_2'' + r''\vartheta + c_1\vartheta'')\\ \hspace{-0.5cm} - (\mathring{w} - c_2(\frac{x_2^2}{2} - \frac{1}{24})) r'' - \ringring{w}_{\langle 0 \rangle} \vartheta''  +(\frac{x_2^2}{2} - \frac{1}{24}) (\vartheta'^2-c_2 r'')|^2 \, dx \end{multlined}\\
			&=\begin{multlined}[t] \int_I |\vartheta'|^2 + \frac12  |\xi_1'+ \frac12 r'^2 + \frac{1}{24}\vartheta'^2|^2 +\frac{J_1}{2} |\xi_2'' + r''\vartheta + c_1\vartheta''|^2 \\
				\qquad+\frac{J_3}{2}  |r''|^2  +\frac{J_4}{2}  |\vartheta''|^2  +\frac{J_2}{2} |c_2 r'' - \vartheta'^2|^2 \, dx_1,
			\end{multlined}
		\end{aligned}
		$$
		where to obtain the last equality we used Lemma \ref{lemmaw0}.

		{\it {(b) (Recovery sequence)}}
		It suffices to consider pairs $(u, w) \in \mathcal{X}$ such that $F^2(u, w) < \infty$. In this case, we have $(u, w) \in \mathcal{A}^{2}$, and the pair can be uniquely represented in terms of functions $(\xi_1, \xi_2, r, \vartheta)$ as introduced in the definition of $\mathcal{A}^{2}$. We therefore define, simply,
		$$ w_\e \coloneqq   \e w, \qquad u_\e \coloneqq    \e^2 u.$$
		Thus
		$$ \begin{aligned}
			\lim_{\e\downarrow 0}	F_{\e}^2(u_\e, w_\e) &= \lim_{\e\downarrow 0}\frac{1}{2} \int_\Omega |\begin{pmatrix}
				\e r'' & \vartheta'\\
				\vartheta' & 0
			\end{pmatrix}|^2 \, dx + \frac12 \int_\Omega |\begin{pmatrix}
				\partial_1 u_1 + \frac12 (\partial_1 w)^2 & 0\\
				0 & 0
			\end{pmatrix}|^2 \, dx\\
			&= \int_I |\vartheta'|^2 \,dx_1 + \frac12\int_\Omega | \partial_1 u_1 + \frac12 (\partial_1 w)^2 |^2 \,dx
		\end{aligned}$$
		as requested.
	\end{proof}

	\subsection{The Euler-Lagrange equations and some examples}
	The presence of nontrivial transverse curvature ($\mathring{w}'' \not\equiv 0$) forces the Euler-Lagrange equations to be of the fourth-order in $r$ and of the fourth-order in $\vartheta$.\\
	Supposing $\xi_1=\xi_2=r=0$ (and $\Lambda_1=\Lambda_2=\Lambda_3=\Phi_2=\Phi_3=0$) the limit energy $F^2$ reduces to, after substituting \eqref{constants},
	$$ \begin{aligned}
		\vartheta \mapsto \int_I |\vartheta'|^2 dx_1 + \frac{1}{640}\int_I |\vartheta'|^4 \, dx_1 +\frac{c_1^2 + 12J_4}{24} \int_I|\vartheta''|^2 \, dx_1
	\end{aligned}$$
	with $\vartheta \in \{\vartheta \in H^2(I) \st \vartheta(-\frac{\ell}{2}) = \vartheta'(-\frac{\ell}{2})=0, \vartheta(\frac{\ell}{2})=\Phi_1, \vartheta'(\frac{\ell}{2})=0 \}$.
	The Euler-Lagrange equation reads
	$$ -(2 + \frac{3}{160}\vartheta'^2)\vartheta'' + \frac{c_1^2 + 12J_4}{12}\vartheta''''=0 \qquad \text{ in } I.$$
	We can interpret this equation as a nonlinear Vlasov torsion model, where the warping stiffness is $\frac{c_1^2 + 12J_4}{12}$ and the torsion stiffness, depending itself by $\vartheta$, is $(2 + \frac{3}{160}\vartheta'^2)$. As expected, the warping stiffness depends $\mathring{w}$.  Note that Vlasov's linear torsion theory has been already deduced by $\Gamma$-convergence from the 3D linear elasticity in \cite{Freddi2007}. Here, we have deduced a nonlinear version of that theory.

	Similarly, supposing $\xi_1=\xi_2=\vartheta=0$ (and $\Lambda_1=\Lambda_2=\Phi_1=\Phi_3=0$) for the purely flexural model, the limit energy reduces to
	$$ r\mapsto \frac18 \int_I |r'^2|^2 \, dx_1 +\frac{720 J_3 +  c_2^2}{1440} \int_I |r''|^2 \, dx_1$$
	with  $r \in \{r \in H^2(I) \st r(-\frac{\ell}{2}) = r'(-\frac{\ell}{2})=0, r(\frac{\ell}{2})=\Lambda_3, r'(\frac{\ell}{2})=\Phi_2\}$.
	The Euler-Lagrange equation reads
	$$ -\frac32r'^2r'' + \frac{720 J_3 +  c_2^2}{720} r''''=0 \qquad \text{ in } I.$$
	Note that in this case the bending stiffness ($\frac{720 J_3 +  c_2^2}{1440}$) is of purely geometrical nature, while the constitutive contribution vanishes at the limit.

	We provide an explicit solution for two cases of interest.
	For $\mathring{w} = \frac{x_2^2}{2} - \frac{1}{24}$ we have
	$$ \ringring{w} = \frac{x_2^3}{6} + \frac{x_2}{24}, \qquad \ringring{w}_{\langle 0 \rangle} = \frac{x_2^3}{6} - \frac{x_2}{40},$$
	$$ c_1 = \frac{1}{15}, \qquad c_2 = 1, \qquad J_3 = 0, \qquad J_4 = \frac{1}{100800}.$$

	While, for $\mathring{w} = (x_2^2 - \frac{1}{16})^2 - \frac{23}{3840}$ we have
	$$ \ringring{w} = \frac{x_2}{480}(288x_2^4 - 20x_2^2 + 1)
	, \qquad \ringring{w}_{\langle 0 \rangle} = \frac{x_2}{480}(288x_2^4 - 20x_2^2 + 1) - \frac{x_2}{84},
	$$
	$$ c_1 = \frac{1}{84}, \qquad c_2 = \frac{5}{28}, \qquad J_3 = \frac{1}{44100}, \qquad J_4 = \frac{757
	}{124185600}.$$

	\section{The $\Gamma$-limit for $0<\beta<2$}\label{sec:gamma1}

	We begin by establishing an auxiliary result that will be used in the proof of the liminf inequality.

	\begin{lemma}\label{lemmarelax}
		Let $(\xi_1,r) \in \calB$, the set defined in Lemma \ref{mulambda}. Then, with the notation of Remark \ref{remarkext},
		\begin{multline}\inf\{\int_{\widetilde{I}} |\widetilde{\xi}_1' + \widetilde{\lambda} + \frac12 \widetilde{r}'^2|^2 \, dx_1 \st  \widetilde{\lambda} \in \calM^+(\widetilde{I}), \ \widetilde{\xi}_1' + \widetilde{\lambda} + \frac12 \widetilde{r}'^2 \in L^2(\widetilde{I})\} \\ = \int_I |(\frac{d\xi_1'}{d\calL}  + \frac12 r'^2)^+|^2 \, dx_1,\end{multline}
		where $\frac{d\xi_1'}{d\calL}$ denotes the Radon-Nikodym derivative of $\xi_1'$ with respect to  the one dimensional Lebesgue measure $\calL$.
	\end{lemma}
	\begin{proof}
		Any $\widetilde{\lambda} \in \calM^+(\widetilde{I})$ can be decomposed uniquely as $\widetilde{\lambda} =  \widetilde{\lambda}_{a} \calL+ \widetilde{\lambda}_s$ where $\widetilde{\lambda}_{a} = \frac{d \widetilde{\lambda}}{d\calL}\in L^1(\widetilde{I})$, and $\widetilde{\lambda}_s$ and $\calL$ are mutually singular. Similarly, we have $\widetilde{\xi}_1' =  (\widetilde{\xi}_1')_a \calL+ (\widetilde{\xi}_1')_s$.

		The requirement $\widetilde{\xi}_1' + \widetilde{\lambda} + \frac12 \widetilde{r}'^2\in L^2(\widetilde{I})$ implies, since $\widetilde{r} \in H^1(\widetilde{I})$, that $(\widetilde{\xi}_1')_s = - \widetilde{\lambda}_s$ and
		$(\widetilde{\xi}_1')_a + \widetilde{\lambda}_a + \frac12 \widetilde{r}'^2\in L^2(\widetilde{I})$. For brevity, we set $\widetilde{f} = (\widetilde{\xi}_1')_a  + \frac12 \widetilde{r}'^2\in L^1(\widetilde{I})$. Let $\widetilde{I}^+:=\{\widetilde{f}\ge 0\}$, $\widetilde{I}^-:=\{\widetilde{f}< 0\} $, and $\widetilde{f}^\pm:=\pm \widetilde{f}\mathbf{1}_{\widetilde{I}^\pm}$.
		We have accordingly
		\begin{align}
			\inf\{\int_{\widetilde{I}} |\widetilde{f}  &+ \widetilde{\lambda}_{a}|^2 \, dx_1 \st \widetilde{\lambda}_{a} \geq 0, \  \widetilde{f}  + \widetilde{\lambda}_{a} \in L^2(\widetilde{I})\}\nonumber\\
			&=\inf\{\int_{\widetilde{I}} |\widetilde{f}^+ - \widetilde{f}^-  + \widetilde{\lambda}_{a}|^2 \, dx_1 \st \widetilde{\lambda}_{a} \geq 0, \  \widetilde{f}  + \widetilde{\lambda}_{a} \in L^2(\widetilde{I})\}\nonumber\\
			&=\inf\{\int_{\widetilde{I}^+} |\widetilde{f}^+ +\widetilde{\lambda}_{a}|^2 \, dx_1+\int_{\widetilde{I}^-} |-\widetilde{f}^-+ \widetilde{\lambda}_{a}|^2 \, dx_1 \st \widetilde{\lambda}_{a} \geq 0, \  \widetilde{f}  + \widetilde{\lambda}_{a} \in L^2(\widetilde{I})\}\nonumber\\
			&\geq\begin{multlined}[t]
				\inf\{\int_{\widetilde{I}^+} |\widetilde{f}^+ + \widetilde{\lambda}_{a}|^2  \, dx_1  \st \widetilde{\lambda}_{a} \geq 0, \  \widetilde{f}^+ + \widetilde{\lambda}_{a} \in L^2(\widetilde{I}^+)\}\\+\inf\{\int_{\widetilde{I}^-} |-\widetilde{f}^- +\widetilde{\lambda}_{a}|^2 \, dx_1  \st \widetilde{\lambda}_{a} \geq 0, \  -\widetilde{f}^-  + \widetilde{\lambda}_{a} \in L^2(\widetilde{I}^-)\}\end{multlined}\nonumber\\
			&= \begin{multlined}[t]
				\inf\{\int_{\widetilde{I}^+} |\widetilde{f}^+ + \widetilde{\lambda}_{a}|^2 \, dx_1  \st \widetilde{\lambda}_{a} \geq 0, \  \widetilde{f}^+  + \widetilde{\lambda}_{a} \in L^2(\widetilde{I}^+)\}
			\end{multlined}\label{eq43}
		\end{align}
		where the last equality has been obtained by taking $\widetilde{\lambda}_{a} = \widetilde{f}^-$ on $\widetilde{I}^-$.
		Since $\widetilde{f}\in L^1(\widetilde{I})$, it follows that $\widetilde{f}^+\in L^1(\widetilde{I}^+)$. Similarly, $\widetilde{\lambda}_a\in L^1(\widetilde{I}^+)$. However, the requirements $\widetilde{f}^+  + \widetilde{\lambda}_{a} \in L^2(\widetilde{I}^+)$ and $\widetilde{f}^+,\widetilde{\lambda}_a\ge 0$, imply that $\widetilde{f}^+, \widetilde{\lambda}_{a} \in L^2(\widetilde{I}^+)$.  Hence, it is allowed to take $\widetilde{\lambda}_a\mathbf{1}_{\widetilde{I}^+}=0$ in \eqref{eq43} and deduce that
		\begin{multline}
			\inf\{\int_{\widetilde{I}} |\widetilde{\xi}_1' + \widetilde{\lambda} + \frac12 \widetilde{r}'^2|^2 \, dx_1 \st  \widetilde{\lambda} \in \calM^+(\widetilde{I}), \ \widetilde{\xi}_1' + \widetilde{\lambda} + \frac12 \widetilde{r}'^2 \in L^2(\widetilde{I})\}\\
			\geq \int_{\widetilde{I}^+} |\widetilde{f}^+|^2 \, dx_1=\int_{\widetilde{I}} |\widetilde{f}^+|^2 \, dx_1.
		\end{multline}

		To show the converse inequality, it suffices to take the admissible choice $\widetilde{\lambda}_{a} = \widetilde{f}^-\mathbf{1}_{\widetilde{I}^-}$, to obtain
		\begin{equation}\label{minlambda}
			\inf\{\int_{\widetilde{I}} |\widetilde{f}  + \widetilde{\lambda}_{a}|^2 \, dx_1 \st \widetilde{\lambda}_{a} \geq 0, \  \widetilde{f}  + \widetilde{\lambda}_{a} \in L^2(\widetilde{I})\}  \leq \int_{\widetilde{I}} |\widetilde{f}^+|^2 \, dx_1.
		\end{equation}
		Henceforth, \eqref{minlambda} holds with the equality sign.
		To conclude the proof, it suffices to notice that by \eqref{ext1} and \eqref{ext3},  $\widetilde{f} = \widetilde{\xi}_1' + \frac12 \widetilde{r}'^2 = 0$ on $\widetilde{I}\setminus I$.
	\end{proof}
	Let $$ \begin{aligned} \calA^{(0,2)} =
		\{(u,w) \in \calX \st
		&\exists \  (\xi_2, \vartheta) \in H^2(I)\times H^1(I), \ \exists \ (\xi_1,r) \in \calB,\\
		&r(-\frac{\ell}{2}) = 0, \ r(\frac{\ell}{2}) = \Lambda_3,\\
		&\vartheta(-\frac{\ell}{2}) = 0, \ \vartheta(\frac{\ell}{2}) = \Phi_1,\\
		&\xi_2(-\frac{\ell}{2})=\xi_2'(-\frac{\ell}{2})=0, \ \xi_2(\frac{\ell}{2})=\Lambda_2, \xi_2'(\frac{\ell}{2})=\Phi_3\\
		&	w = r, u_1 = \xi_1 - x_2\xi_2', \ u_2 = \xi_2\}.
	\end{aligned}$$
	We note that given $(u,w) \in \calA^{(0,2)}$, we uniquely determine $r$ and $\xi_2$, while $\xi_1$ is found up to a constant. The following
	functional on $\calX$ is well-defined:
	$$ F^{(0,2)}(u,w) = \begin{cases}
		\begin{multlined}[t]
			\frac12\int_I 2|\vartheta'|^2 + J_1|\xi_2''|^2  +  |(\frac{d\xi_1'}{d\calL} + \frac12 r'^2)^+|^2 \, dx_1
		\end{multlined} & \text{ if } (u,w)\in\calA^{(0,2)}\\
		+\infty & \text{ otherwise}.
	\end{cases}$$

	We now state and prove the $\Gamma$-convergence result for $0<\beta < 2$. Again, we express the convergence using strong convergence in $\mathcal{X}$, even though this can be significantly improved.
	\begin{theorem}\label{thm:gc222}
		As $\e\downarrow 0$, the sequence of functionals $(F_{\e}^{(0,2)})$ $\Gamma$-converges to $F^{(0,2)}$ in $\calX$, in the following sense:
		\begin{enumerate}[(a)]
			\item (Liminf inequality) for every $(u,w)\in \calX$ and for every sequence $(u_\e, w_\e)\subset \calX$ such that
			$$
			(\frac{u_\e}{\e^\beta},\frac{w_\e}{\e^{\beta/2}})\to (u,w)\quad \mbox{in }\calX
			$$
			we have
			\begin{equation*}
				\liminf\limits_{\e\downarrow 0} F_{\e}^{(0,2)}(u_\e, w_\e)\geq F^{(0,2)}(u,w);
			\end{equation*}
			\item (Recovery sequence)
			for every  $(u,w)\in \calX$ there exists a sequence $(u_\e, w_\e)\subset \calX$ such that
			$$
			(\frac{u_\e}{\e^\beta},\frac{w_\e}{\e^{\beta/2}})\to (u,w)\quad \mbox{in }\calX
			$$
			and
			\begin{equation*}
				\lim\limits_{\e\downarrow 0} F_{\e}^{(0,2)}(u_\e, w_\e) = F^{(0,2)}(u,w).
			\end{equation*}
		\end{enumerate}
	\end{theorem}
	\begin{proof}
		{\it {(a) (Liminf inequality)}}
		Without loss of generality, we assume $\liminf_{\e\downarrow 0} F_{\e}^{(0,2)}(u_\e, w_\e) <\infty$. Up to a subsequence, $\sup_{\e}F_{\e}^{(0,2)}(u_\e, w_\e) <\infty$ and Lemma \ref{compactness4} holds.  With $S_\e$ as defined in \eqref{defSe} we have
		\begin{eqnarray*}
			\liminf_{\e\downarrow 0} F_{\e}^{(0,2)}(u_\e,w_\e)&\ge&
			\liminf_{\e\downarrow 0}\frac{1}{2}\int_\Omega \big| \nabla^2_\e w_\e- \frac{\mathring{w} ''}{\e} \mathsf{e}_2 \otimes \mathsf{e}_2\big|^2 +|S_\e|^2\, dx\\
			&\ge&
			\liminf_{\e\downarrow 0}\int_\Omega (\nabla^2_\e w_\e)_{12}^2+\frac{1}{2} (S^\e_{11})^2\, dx\\
			&=&
			\liminf_{\e\downarrow 0}(\int_\Omega (\nabla^2_\e w_\e)_{12}^2 \, dx+ \int_{\widetilde{\Omega}}\frac{1}{2} (\widetilde{S}^\e_{11})^2\, dx)\\
			&\ge&\int_\Omega |\vartheta'|^2 \, dx_1+\frac{1}{2} \int_{\widetilde{\Omega}}|\widetilde{\xi}_1' + \widetilde{\lambda} + \frac12 \widetilde{r}'^2 - x_2\widetilde{\xi}_2''|^2\, dx
		\end{eqnarray*}
		where the last inequality follows from \eqref{cmp2} and Lemma \ref{mulambda} with $\widetilde{\lambda} \in \calM^+(\widetilde{I})$. Integrating we find
		\begin{eqnarray*}
			\liminf_{\e\downarrow 0} F_{\e}^{(0,2)}(u_\e,w_\e)
			&\ge& \int_I |\vartheta'|^2 \ dx_1 + \frac12 \int_{\widetilde{I}}| \widetilde{\xi}_1' + \widetilde{\lambda} + \frac12 \widetilde{r}'^2 |^2 \,dx_1 + \frac{J_1}{2}\int_{\widetilde{I}} | \widetilde{\xi}_2'' |^2 \,dx_1\\
			&\ge& \begin{multlined}[t]\int_I |\vartheta'|^2 \,dx_1 + \frac{J_1}{2}\int_{\widetilde{I}} | \widetilde{\xi}_2'' |^2 \,dx_1 \\ + \begin{multlined}[t]\frac12\inf\{\int_{\widetilde{I}} | \widetilde{\xi}_1' + \widetilde{\lambda} + \frac12 \widetilde{r}'^2 |^2 \,dx_1 \st \widetilde{\lambda} \in \calM^+(\widetilde{I}), \\ \widetilde{\xi}_1' + \widetilde{\lambda} + \frac12 \widetilde{r}'^2 \in L^2(\widetilde{I}) \}\end{multlined} \end{multlined}\\
			&=& \int_I |\vartheta'|^2 \,dx_1 + \frac{J_1}{2}\int_I | \xi_2'' |^2 \,dx_1 + \frac12\int_I | (\frac{d\xi_1'}{d\calL} + \frac12 r'^2)^+ |^2 \,dx_1,
		\end{eqnarray*}
		where the last equality follows from Lemma \ref{lemmarelax} and \eqref{ext2}.

		{\it {(b) (Recovery sequence)}} We adapt some arguments presented in \cite{Conti2005}.

		Without loss of generality, let us assume $F^{(0,2)} <\infty$. Accordingly, $(u,w) \in \calA^{(0,2)}$. Hence, there exists $(\xi_2, \vartheta) \in H^2(I) \times H^1(I)$, $(\xi_1, r) \in \calB$  such that $w = r, u_1 = \xi_1' - x_2\xi_2''$, and $\ u_2 = \xi_2$.

		\textit{Step 1.} Let $\eta_\delta$ be the standard mollifier and $(\xi_{1\delta}, \xi_{2\delta}, r_{\delta}, \vartheta_\delta) =\eta_\delta \ast (\xi_1, \xi_2, r, \vartheta) $. Let
		$u_\delta=( \xi_{1\delta}' - x_2\xi_{2\delta}'',\xi_{2\delta})$ and $w_\delta=r_\delta$.
		Assume for the time being that, for every fixed $\delta$, we can find a recovery sequence, i.e., a sequence $(u_{\delta,\e}, w_{\delta, \e}) \subset \calX$ such that
		$$
		(\frac{u_{\delta,\e},}{\e^\beta},\frac{w_{\delta,\e},}{\e^{\beta/2}})\to (u_\delta,w_\delta)\quad \mbox{in }\calX
		$$
		as $\e\downarrow 0$, and
		$$ \limsup_{\e\downarrow 0} F_{\e}^{(0,2)}(u_{\delta,\e}, w_{\delta, \e}) \leq F^{(0,2)}(u_{\delta}, w_{\delta}).$$
		Note that $\vartheta_\delta\to\vartheta$ in $H^1(I)$ and also $\xi_{2\delta}\to\xi_2$ in $H^2(I)$.	Then 		\begin{align}
			\limsup_{\delta\downarrow 0}&\limsup_{\e\downarrow 0} F_{\e}^{(0,2)}(u_{\delta,\e}, w_{\delta, \e}) \leq \limsup_{\delta\downarrow 0} F^{(0,2)}(u_{\delta}, w_{\delta})\nonumber\\
			&= \limsup_{\delta\downarrow 0} \int_I |\vartheta_\delta'|^2 \,dx_1 + \frac{J_1}{2}\int_I | \xi_{2,\delta}'' |^2 \,dx_1 + \frac12\int_I | (\xi_{1\delta}' + \frac12 (r_\delta')^2)^+ |^2 \,dx_1\nonumber\\
			&=  \int_I |\vartheta'|^2 \,dx_1 + \frac{J_1}{2}\int_I | \xi_{2}'' |^2 \,dx_1 + \frac12 \limsup_{\delta\downarrow 0}\int_I | (\xi_{1\delta}' + \frac12 (r_\delta')^2)^+ |^2 \,dx_1.\label{rt1}
		\end{align}
		Let $\widetilde{\lambda} \in \calM^+(\widetilde{I})$ such that $\widetilde{\xi}_1' + \widetilde{\lambda} + \frac12 \widetilde{r}'^2 \in L^2(\widetilde{I})$. Let $\widetilde{\lambda}_\delta=\eta_\delta\ast \widetilde{\lambda}$ and note that $\widetilde{\lambda}_\delta \ge 0$ almost everywhere on $\widetilde{I}$. Using
		the nondecreasing property of the map
		$a\mapsto (a + \frac12 b^2)^+$, for every $b$, we deduce that
		\begin{equation}\label{rt2}
			(\xi_{1\delta}' + \frac12 (r_\delta')^2)^+\le (\xi_{1\delta}' +\widetilde{\lambda}_\delta+ \frac12 (r_\delta')^2)^+ \qquad \text{on} \ I.
		\end{equation}

		By Jensen's inequality, for every  function $z$ and every convex function $g$ we have that
		$g(\eta_\delta \ast z)\le \eta_\delta\ast \big(g(z)\big)$. By means of this inequality and observing that the map $(a,b)\mapsto (a + \frac12 b^2)^+$ is convex, we have that
		\begin{equation}\label{rt3}(\xi_{1\delta}' +\widetilde{\lambda}_\delta+ \frac12 (r_\delta')^2)^+
			\le \eta_\delta\ast\big((\xi_1' +\widetilde{\lambda}+ \frac12 (r')^2)^+\big) \qquad \text{on} \ I.
		\end{equation}
		Since $\widetilde{\xi}_1' + \widetilde{\lambda} + \frac12 \widetilde{r}'^2 \in L^2(\widetilde{I})$ also $(\xi_1' +\widetilde{\lambda}+ \frac12 (r')^2)^+\in L^2(I)$. Therefore $\eta_\delta\ast\big((\xi_1' +\widetilde{\lambda}+ \frac12 (r')^2)^+\big)\to (\xi_1' +\widetilde{\lambda}+ \frac12 (r')^2)^+$ in $L^2(I)$.
		Thus, by using \eqref{rt1}, \eqref{rt2}, \eqref{rt3}, \eqref{ext1}, and \eqref{ext3} we obtain
		$$ \begin{aligned}
			\limsup_{\delta\downarrow 0}&\limsup_{\e\downarrow 0} F_{\e}^{(0,2)}(u_{\delta,\e}, w_{\delta, \e}) \\
			&\le  \int_I |\vartheta'|^2 \,dx_1 + \frac{J_1}{2}\int_I | \xi_{2}'' |^2 \,dx_1 + \frac12 \limsup_{\delta\downarrow 0}\int_I | \eta_\delta\ast\big((\xi_1' +\widetilde{\lambda}+ \frac12 (r')^2)^+\big)|^2 \,dx_1\\
			&=  \int_I |\vartheta'|^2 \,dx_1 + \frac{J_1}{2}\int_I | \xi_{2}'' |^2 \,dx_1 + \frac12 \int_I | (\xi_1' +\widetilde{\lambda}+ \frac12 (r')^2)^+|^2 \,dx_1\\
			&\le\int_I |\vartheta'|^2 \,dx_1 + \frac{J_1}{2}\int_I | \xi_{2}'' |^2 \,dx_1 + \frac12\int_{\widetilde{I}} | \widetilde{\xi}_{1}' + \widetilde{\lambda} + \frac12 \widetilde{r}'^2 |^2 \,dx_1.
		\end{aligned}$$
		Now, by invoking Lemma \ref{lemmarelax}, we deduce that
		$$ \begin{aligned}
			\limsup_{\delta\downarrow 0}&\limsup_{\e\downarrow 0} F_{\e}^{(0,2)}(u_{\delta,\e}, w_{\delta, \e}) \\
			&\le  \int_I |\vartheta'|^2 \,dx_1 + \frac{J_1}{2}\int_I | \xi_{2}'' |^2 \,dx_1 + \frac12 \int_I | (\frac{d\xi_1'}{d\calL} +\frac12 (r')^2)^+|^2 \,dx_1.		\end{aligned}$$
		By a diagonalization argument (see for instance \cite[Corollary 1.16]{Attouch}), we can find a map $\e \mapsto \delta(\e)$, with $\delta(\e)\downarrow 0$ as $\e\downarrow 0$, such that
		$$
		(\frac{u_{\delta(\e),\e}}{\e^\beta},\frac{w_{\delta(\e),\e}}{\e^{\beta/2}})\to (u,w)\quad \mbox{in }\calX
		$$
		and
		\begin{multline}\limsup_{\e\downarrow 0} F_{\e}^{(0,2)}(u_{\delta(\e),\e}, w_{\delta(\e), \e}) \\ \leq \int_I |\vartheta'|^2 \,dx_1 + \frac{J_1}{2}\int_I | \xi_{2}'' |^2 \,dx_1 + \frac12 \int_I | (\frac{d\xi_1'}{d\calL} +\frac12 (r')^2)^+|^2 \,dx_1.\end{multline}
		Hence, $(u_{\delta(\e),\e}, w_{\delta(\e), \e})$ is the requested recovery sequence.
		It remains to verify that a recovery sequence exists for $u$ and $w$ smooth.

		\textit{Step 2.}
		Assume $(u,w) \in \calA^{(0,2)}\cap (C^\infty(\overline{\Omega}, \bbR^2)\times C^\infty(\overline{\Omega}))$. Accordingly, there are $(r,\vartheta, \xi_1, \xi_2) \in (C^\infty(\overline{I}))^4$ as in the definition of $\calA^{(0,2)}$ (apart from regularity).

		Consider now the extension $\widetilde{\xi}_1$ as in \eqref{ext1}.
		Despite $\xi_1 \in C^\infty(\overline{I})$, $\widetilde{\xi}_1 \in BV(\widetilde{I})$ with (possible) jump discontinuities at $x_1 = \pm \ell/2$. We recall that $(\xi_1, r) \in \calB$: there exists $\widetilde{\lambda} \in \calM^+(\widetilde{I})$ such that
		$\widetilde{\xi}_1' + \widetilde{\lambda} + \frac12\widetilde{r}'^2 \in L^2(\widetilde{I})$. This and \eqref{ext1} imply that $\Lambda_1 - \xi_1(\ell/2)\leq 0$ and $\xi_1(-\ell/2)\leq 0$. As previously noticed, $\xi_1$ is defined up to a constant: we may therefore consider (up to a translation) that $\xi_1(-\ell/2)=0$.
		Set $p^2 := \xi_1(\ell/2) - \Lambda_1$. Fix $n>0$ and define
		$$ \xi_{1n}(x_1)\coloneqq
		\xi_1(x_1) - n p^2 (x_1 - \frac{\ell}{2} + \frac{1}{n}) \mathbf{1}_{(\frac{\ell}{2} - \frac{1}{n}, \frac{\ell}{2})}(x_1).$$
		One can check that $\xi_{1n} \weakstar \xi_1$ in $BV(I)$ as $n\uparrow 0$.
		Let $m_n \in L^\infty(I, [0,\infty))$ be a simple function that takes constant value on the disjoint segments $I_j = (-\frac{\ell}{2} + (j-1)\frac{\ell}{n}, -\frac{\ell}{2} + j\frac{\ell}{n}]$ covering $I$ and such that
		\begin{equation}\label{Cn}\int_I |\xi_1' + m_n + \frac12 r'^2|^2 \, dx_1 \leq  \int_I |(\xi_1' + \frac12 r'^2)^+|^2 \, dx_1 + \frac{C}{n}.\end{equation}
		Let $\zeta_\e \subset H^2_0(I)$ be a sequence such that $\zeta_\e \weak 0$ in $H^1(I)$, $\e^{\min\{\beta/2,1-\beta/2\}}\zeta''_\e \to 0$ in $L^2(I)$, $\frac12(\zeta'_\e)^2 \to m_n + n p^2\mathbf{1}_{(\frac{\ell}{2}-\frac{1}{n}, \frac{\ell}{2})}$ in $L^2(I)$. Such a sequence exists by Lemma \ref{lemmazeta}, since $m_n + n p^2\mathbf{1}_{(\frac{\ell}{2}-\frac{1}{n}, \frac{\ell}{2})}$ is a positive simple function by its own.\\
		From Lemma \ref{lemmaa7}, there exist functions $g_\e, h_\e \in H^2(I)$ such that
		\begin{align*}
			&g_\e(-\ell/2) = g_\e(\ell/2) = 0, & g_\e'(-\ell/2)= -r'(-\ell/2), & \quad g_\e'(\ell/2)=\Phi_2 - r'(\ell/2) & \forall \e > 0,\\
			&h_\e(-\ell/2) = h_\e(\ell/2) = 0, & h_\e'(-\ell/2)= -\vartheta'(-\ell/2), & \quad h_\e'(\ell/2)=-\vartheta'(\ell/2) & \forall \e > 0,
		\end{align*}
		and
		$$ g_\e, h_\e \to 0 \quad in \ W^{1,4}(I), \qquad \e^{\min\{\beta/2,1-\beta/2\}} g_\e'', \e^{\min\{\beta/2,1-\beta/2\}} h_\e'' \to 0 \ in \ L^2(I). $$
		Define the recovery sequence as
		$$ \begin{aligned}
			u_{1\e}(x) &\coloneqq  \begin{multlined}[t]
				\e^\beta(\xi_{1n} - x_2\xi_2') - \e^{1+\beta/2} (r' + g_\e'+ \zeta_\e')(\mathring{w} + x_2 (\vartheta + h_\e))
				-\e^{2}\ringring{w}(\vartheta' + h_\e')\\ - \e^\beta\int_{-\ell/2}^{x_1}\zeta_\e'(t) (r'(t) + g_\e'(t)) \, dt + \e^\beta(x_1 + \ell/2) \Rint_{-\ell/2}^{\ell/2}\zeta_\e'(t) (r'(t) + g_\e') \, dt
			\end{multlined}\\
			u_{2\e}(x) &\coloneqq  \e^\beta\xi_2- \e^{2}x_2 \frac{(\vartheta + h_\e)^2}{2} - \e^2\mathring{w}(\vartheta + h_\e) \\
			w_\e(x) &\coloneqq  \e^{\beta/2}(r + g_\e + \zeta_\e) + \e x_2 (\vartheta + h_\e) + \e\mathring{w}.
		\end{aligned} $$
		where $\ringring{w}$ has been defined in \eqref{w0}.
		Note that $\int_{-\ell/2}^{x_1}\zeta_\e'(t) r'(t) \, dt\weak 0$ in $W^{1,1}(I)$. In fact $\zeta_\e'r' \weak 0$ in $L^1(I)$, and
		$$ \begin{aligned}
			\int_I |\int_{-\ell/2}^{x_1}\zeta_\e'(t) r'(t) \, dt|\, dx_1 &= \int_I |\int_I \mathbf{1}_{(-\ell/2, x_1)} \zeta_\e'(t) r'(t)\, dt | \, dx_1 \\\ &= (\int_I \mathbf{1}_{(-\ell/2, x_1)} dx_1) |\int_I \zeta_\e'(t) r'(t)\, dt| \to 0.
		\end{aligned}$$

		Accordingly, we have $S_{12}^\e = S_{22}^\e = 0$,
		$$ S_{11}^\e = \begin{multlined}[t]
			\xi_{1n}' - x_2 \xi_2'' - \e^{1-\beta/2}(r''+g_\e'' + \zeta_\e'')(\mathring{w} + x_2(\vartheta + h_\e)) \\- \e^{2-\beta}\ringring{w}(\vartheta'' + h_\e'') + \Rint_I \zeta_\e'(r' + g_\e') \, dx_1 + \frac12(r' + g_\e')^2 \\ + \frac12 (\zeta_\e')^2 + \frac12 \e^{1-\beta/2}x_2 (\vartheta' + h_\e')^2,
		\end{multlined}$$
		and
		$$ S_{11}^\e \to \xi_1' + m_n + \frac12 r'^2 - x_2 \xi_2'' \qquad \text{in} \ L^2(\Omega).$$
		Passing to the limit we deduce
		$$ \begin{aligned}
			\lim_{\e\downarrow 0} F_{\e}^{(0,2)}(u_\e, w_\e) &= \lim_{\e\downarrow 0} \begin{multlined}[t]\frac{1}{2} \int_\Omega |\begin{pmatrix}
					\e^{\beta/2} (r'' + g_\e'' + \zeta_\e'') + \e x_2(\vartheta'' + h_\e'') & \vartheta' + h_\e' \\
					sym & 0
				\end{pmatrix}|^2 \, dx \\
				+ \frac12 \int_\Omega |S_{11}^\e|^2 \, dx \end{multlined}\\
			&= \int_I |\vartheta'|^2 \, dx_1 + \frac12 \int_I |\xi_1' + m_n + \frac12 r'^2|^2 \, dx_1 + \frac{J_1}{2} \int_I |\xi_2''|^2 \, dx_1\\
			&\leq \int_I |\vartheta'|^2 \,dx_1 + \frac{J_1}{2}\int_I | \xi_2'' |^2 \,dx_1 + \frac12\int_I | (\xi_1' + \frac12 r'^2)^+ |^2 \,dx_1 + \frac{C}{n}
		\end{aligned} $$
		where we used \eqref{Cn} to deduce the last inequality.
		Since $n$ is arbitrary, the recovery sequence condition is proved.						
	\end{proof}

	\begin{remark}
		In the regime $0<\beta<2$, the limit model does not see the transverse curvature $\mathring{w}''$.
		The limit model is essentially an elastic string, since it cannot sustain contractions by the well-known effect of relaxation (see \cite{Conti2005, Acerbi1991a, LeDret1995}). However, the string can sustain (uniform) torsion and in-plane bending.
	\end{remark}

	\begin{remark}
		By changing variable, $\zeta \coloneqq   \frac{w}{\e^{\beta/2}}$ and $\eta \coloneqq   \frac{u}{\e^\beta}$, after little manipulation, \eqref{cond} rewrites as
		\begin{equation}\label{condresc}
			\begin{aligned}
				&	\eta(-\frac{\ell}{2}, x_2) = (0,0), \qquad \zeta(-\frac{\ell}{2}, x_2) = \e^{1-\beta/2}\mathring{w}, \qquad \partial_1 \zeta(-\frac{\ell}{2}, \cdot)=0, \\
				&	\eta(\frac{\ell}{2}, x_2) =
				(\Lambda_1 -\e^{1-\beta/2}\Phi_2\mathring{w} -\e^{1-\beta/2} \Phi_2\Phi_1 x_2 - \Phi_3 x_2, \ \Lambda_2 - \e^{2-\beta} \Phi_1\mathring{w} -\e^{2-\beta} \frac{\Phi_1^2 x_2}{2}),\\
				&	\zeta(\frac{\ell}{2}, x_2) = \e^{1-\beta/2}\mathring{w} + \Lambda_3 + \e^{1-\beta/2} x_2 \Phi_1, \qquad   \partial_1 \zeta(\frac{\ell}{2}, x_2)= \Phi_2,
			\end{aligned}
		\end{equation}
		and we have actually found out the $\Gamma$-limit for the family of functionals on $\calX$
		$$ (\eta, \zeta) \mapsto \begin{cases}
			\begin{multlined}[t]
				\e^\beta \frac12\int_\Omega |\nabla^2_\e \zeta - \frac{\mathring{w}''}{\e^{1+\beta/2}} \mathsf{e}_2 \otimes \mathsf{e}_2|^2 \, dx \\+ \frac12 \int_\Omega |E_\e \eta + \frac12 \nabla_\e \zeta\otimes \nabla_\e \zeta - \frac12 \frac{1}{\e^\beta} \mathring{w}'^2 \mathsf{e}_2 \otimes \mathsf{e}_2 |^2 \, dx
			\end{multlined} & \text{ if } \begin{multlined}[t]
				(\eta, \zeta) \in \calA,
			\end{multlined}\\
			+\infty & \text{ otherwise }
		\end{cases}$$
		for $0<\beta \leq 2$, where $\calA := \{(\eta, \zeta)\in H^1(\Omega,\bbR^2)\times H^2(\Omega), \st \eqref{condresc} \text{ holds}\}$.
	\end{remark}

	\appendix

	\section{Technical results}
	We provide here some technical results. For some of them, we claim no originality. Nevertheless, we provide a proof for the sake of completeness.
	If not specified, $I_1, I_2 \subset \bbR$ are generic bounded intervals.

	The following lemma is a well-known version for distributions of the Du Bois-Reymond lemma.
	\begin{lemma}\label{lemmadistr}
		Let $\Omega = I_1 \times I_2$, $I_2 = (a,b)$, and $u \in \calD'(\Omega)$ with $\partial_2 u = 0$ in the sense of distributions. Then, there exists $g \in \calD'(I_1)$ such that $u=g\otimes 1$.
	\end{lemma}
	\begin{proof}
		Let $\varphi \in \calD(\Omega)$ and $\psi \in \calD(I_2)$ with $\int_{I_2} \psi \, dx_2 = 1$.
		Define
		$$ \eta(x_1,x_2) \coloneqq   \int_a^{x_2} \varphi(x_1,s) - (\int_{I_2}\varphi(x_1,t)\, dt) \psi(s) \, ds$$
		which clearly belongs to $\calD(\Omega)$.
		Hence, denoting with $\langle \cdot, \cdot\rangle_X$ the pairing $\langle \cdot , \cdot\rangle_{\calD'(X)\times \calD(X)}$
		$$ \begin{aligned}
			\langle u, \varphi \rangle_{I_1\times I_2}
			&=\langle u, \partial_2 \eta + (\int_{I_2}\varphi(x_1,t)\, dt) \psi(x_2)\rangle_{I_1\times I_2}\\
			&= \langle u, (\int_{I_2}\varphi(x_1,t)\, dt) \psi(x_2)\rangle_{I_1\times I_2}
		\end{aligned}$$
		since $\langle u, \partial_2 \eta \rangle = 0$.
		For fixed $\psi$, the map $\zeta(x_1) \mapsto \langle u, \zeta(x_1)\psi(x_2)\rangle $ defines a distribution $g\in\calD'(I_1)$, i.e.
		$$ \langle g, \zeta\rangle_{I_1} \coloneqq   \langle u, \zeta \psi \rangle_{I_1\times I_2} \qquad \forall \zeta \in \calD(I_1).$$
		Hence (see also \cite[Ch. 2.7]{Mitrea2013})
		$$ \begin{aligned}
			\langle u , \varphi \rangle_{I_1\times I_2} = \langle g, \int_{I_2}\varphi(x_1,t)\, dt \rangle_{I_1}=
			\langle g, \langle 1, \varphi\rangle_{I_2} \rangle_{I_1}=
			\langle g\otimes 1, \varphi\rangle_{I_1\times I_2}.
		\end{aligned} $$
	\end{proof}

	Next we specialize	Lemma \ref{lemmadistr} for $L^p$ functions and finite Radon measures.
	\begin{corollary}\label{corLp}
		\begin{enumerate}[(i)]
			\item If $u \in \calM(\Omega)$ with $\partial_2 u = 0$, there exists $g \in \calM(I_1)$ such that $u = g \otimes \calL$ on $\Omega$.
			\item If $u \in L^p(\Omega)$ ($1<p<\infty$) with $\partial_2 u = 0$, there exists $g \in L^p(I_1)$ such that $u = g$ a.e on $\Omega$.
		\end{enumerate}
	\end{corollary}
	\begin{proof}
		Claim	\textit{(i)}. By Lemma \ref{lemmadistr} exists $g \in \calD'(I_1)$ such that
		and for any $\varphi_i \in \calD(I_i)$ with $\int_{I_2} \varphi_2 = 1$ it holds that
		$$ \int_{\Omega} \varphi_1 \varphi_2 \, du = \langle g, \varphi_1 \rangle_{I_1}.$$
		We have
		$$ |\langle g, \varphi_1\rangle|_{I_1} = |\langle u, \varphi_1\varphi_2 \rangle| \leq C|u|(\Omega) \n{\varphi_1}_{L^\infty} <\infty. $$
		By density, $g$ is a continuous linear functional on $C_0(I_1)$ and by Riesz's  representation theorem there exists a unique element of $\calM(I_1)$, still denoted by $g$, such that
		$$ \langle g, \varphi_1\rangle = \int_{I_1} \varphi_1 \, dg \qquad \forall \varphi_1 \in C_0(I_1).$$
		Moreover, for every $\varphi \in C_0(\Omega)$,
		$$ \int_{\Omega}\varphi \, du = \int_{I_1} \int_{I_2}\varphi \, dx_2 dg  = \int_{\Omega} \varphi \, d(g \otimes \calL)$$
		and we conclude $u = g \otimes \calL$.

		Claim	\textit{(ii)}. It follows from claim (i) when considering measures of the form $u\calL^2$ with $u \in L^p(\Omega)$ and the $L^p$ version of Riesz's representation theorem.
	\end{proof}

	\begin{lemma}\label{lemmabd}
		Let $\Omega = I_1\times I_2$ and $u \in BD(\Omega,\bbR^2)$ with $\partial_2 u_2 = 0$ and $\partial_1 u_2 + \partial_2 u_1 = 0$. Then, there exist $(\xi_1, \xi_2) \in BV(I_1) \times BH(I_1)$, such that
		$$ u_{1} = \xi_1 -x_2\xi_2',\qquad
		u_2 = \xi_2.$$
	\end{lemma}
	\begin{proof}
		By Corollary \ref{corLp} and by the embedding $BD(\Omega,\bbR^2)\embed L^2(\Omega,\bbR^2)$, there exists $\xi_2 \in L^2(I_1)$ such that
		$u_2  = \xi_2 $.
		Take now $\varphi \in C_c^\infty(\Omega)$.
		$$ \begin{aligned}
			0 &= \int_\Omega \varphi \, d (\partial_1 u_2 + \partial_2 u_1)
			=-\int_\Omega \partial_1\varphi u_2 + \partial_2 \varphi u_1 \, dx\\
			&=-\int_\Omega \partial_1\varphi \xi_2 + \partial_2 \varphi u_1 \, dx
			=-\int_\Omega \partial_1\varphi \partial_2(x_2\xi_2) + \partial_2 \varphi u_1 \, dx\\
			&=\int_\Omega \partial_{12}\varphi (x_2\xi_2) - \partial_2 \varphi u_1 \, dx
			= -\langle \partial_1 (\xi_2 \otimes x_2) + u_1,\partial_2 \varphi \rangle.
		\end{aligned}$$
		By Lemma \ref{lemmadistr}, there exists $\xi_1 \in \calD'(I_1)$ such that
		$
		\xi_1 \otimes 1 = \xi_2' \otimes x_2 + u_1
		$
		from which
		$$
		u_1 = \xi_1 \otimes 1 - \xi_2' \otimes x_2.
		$$
		We now show that $(\xi_1, \xi_2) \in BV(I_1)\times BH(I_1)$. For every
		$\varphi_\alpha \in C_c^\infty(I_\alpha)$ with $\varphi_2$ even and such that $\int_{I_2} \varphi_2 \, dx_2 = 1$ it holds
		\begin{equation*}
			\int_\Omega u_1 \varphi_1 \varphi_2 \, dx = \langle \xi_1, \varphi_1 \rangle, \qquad
			\int_\Omega  \varphi_1 \varphi_2 \, d(\partial_1 u_1) = \langle \xi_1', \varphi_1 \rangle. \end{equation*}
		By the continuous embedding $BD(\Omega, \bbR^2) \embed L^2(\Omega, \bbR^2)$ we obtain
		$$|\langle \xi_1, \varphi_1 \rangle| \leq C \n{u_1}_{L^2(\Omega)} \n{\varphi_1}_{L^2(I_1)}, \qquad |\langle \xi_1', \varphi_1 \rangle| \leq C |\partial_1u_1|(\Omega)\n{\varphi_1}_{C^0(I_1)},$$
		and by the Riesz's representation theorem we conclude $\xi_1\in BV(I_1)$.
		With similar arguments, but working with $\varphi_2$ odd, we deduce $\xi_2 \in BH(I_1)$.
	\end{proof}

	\begin{lemma}\label{lemmaext}
		Let $\Omega = I_1 \times I_2$ and $f \in H^2(\Omega)$. There exist a function $\hat{f} \in H^2(\bbR^2)$ such that $\hat{f}|_{\Omega} = f$ and a positive constant $C$, independent of $f$, such that $\|\partial_{12}\hat{f}\|_{L^2(\bbR^2)}\leq C (\|\partial_{12} f\|_{L^2(\Omega)} + \n{f}_{H^1(\Omega)})$.
	\end{lemma}
	\begin{proof}
		Let $f \in H^2(\Omega)$.
		Up to a change of reference frame, we may suppose $I_1 = (0, a)$, $I_2 = (0,b)$ for some $a,b > 0$.
		We start by extending $f$ on the domain obtained by the union of $\Omega$ and its reflection along $x_2 = 0$. Let us call this domain $\Omega_1$.
		Define  (see \cite[Ch. 2, Theorem 3.9]{Necas2012})
		$$ f_1(x_1, x_2) \coloneqq   \begin{cases}
			f(x_1, x_2) & x \in \Omega\\
			-3f(x_1, -x_2) +4 f(x_1, -\frac12 x_2) & x \in \Omega_1\setminus\Omega.
		\end{cases}$$
		It is clear that $f_1 \in H^2(\Omega_1)$. Moreover
		$$ \n{f_1}_{L^2(\Omega_1)} \leq C\n{f}_{L^2(\Omega)}, \quad \n{\nabla f_1}_{L^2(\Omega_1)} \leq C\n{\nabla f}_{L^2(\Omega)},$$ $$\n{\partial_{12} f_1}_{L^2(\Omega_1)} \leq C\n{\partial_{12} f}_{L^2(\Omega)}. $$
		Iterating properly three times more, mirroring along axes parallel to the coordinates axes, we end up with a function $f_4$ defined on a rectangle $\Omega_4$ that compactly contains $\Omega$ (see Fig. \ref{fig:ext}).
		\begin{figure}[!hbt]\label{fig:ext}
			\centering
			\includegraphics[width=\textwidth]{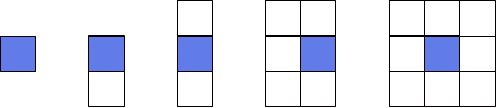}
			\caption{A series of domain reflections. $\Omega$ in blue.}
		\end{figure}

		Moreover, $$\n{f_4}_{L^2(\Omega_4)} \leq C\n{f}_{L^2(\Omega)}, \quad \n{\nabla f_4}_{L^2(\Omega_4)} \leq C\n{\nabla f}_{L^2(\Omega)}, $$
		$$\n{\partial_{12} f_4}_{L^2(\Omega_4)} \leq C\n{\partial_{12} f}_{L^2(\Omega)}.$$

		We now take a cutoff function $\varphi \in C_c^\infty(\Omega_4, [0,1])$ such that $\varphi=1$ on $\Omega$ and pose $\hat{f} \coloneqq   \varphi f_4$ extended by zero to the whole $\bbR^2$.
		Thus
		$$\begin{aligned}
			\n{\partial_{12}\hat{f}}_{L^2(\bbR^2)} &= \n{\partial_{12}\hat{f}}_{L^2(\Omega_4)}\\
			&\leq C(\n{\partial_{12}f_4\varphi}_{L^2(\Omega_4)}
			+ \n{f_4\partial_{12}\varphi}_{L^2(\Omega_4)} +\n{\partial_2 f_4\partial_{1}\varphi}_{L^2(\Omega_4)}\\
			&\hspace{7cm}+\n{\partial_1 f_4\partial_{2}\varphi}_{L^2(\Omega_4)})\\
			& \leq C (\n{\partial_{12}f_4}_{L^2(\Omega_4)}
			+ \n{f_4}_{L^2(\Omega_4)} +\n{\partial_2 f_4}_{L^2(\Omega_4)}
			+\n{\partial_1 f_4}_{L^2(\Omega_4)})\\
			& \leq C (\n{\partial_{12}f}_{L^2(\Omega)}
			+ \n{f}_{L^2(\Omega)} +\n{\partial_2 f}_{L^2(\Omega)}
			+\n{\partial_1 f}_{L^2(\Omega)}).
		\end{aligned}$$
	\end{proof}

	\begin{lemma}\label{lemmader2}
		Let $f \in H^1(I)$. There exists a sequence $(f_n) \subset C^\infty(\overline{I})$ such that $f_n \to f$ in $H^1(I)$ as $n \uparrow \infty$. Moreover, if $(\e_n)\subset \bbR$ is such that $\e_n \downarrow 0$ as $n \uparrow \infty$,  $(f_n)$ can be chosen such that
		$\e_n f_n''\to 0$ in $L^2(I)$.
	\end{lemma}
	\begin{proof}
		The proof is a refinement of \cite[Lemma 11]{Paroni2024c}.
		Set $\widetilde I\coloneqq  (-\frac32\ell,\frac32\ell)$ and let $\widetilde{f} \in H^1(\widetilde{I})$ be an extension of $f$.
		Take $\varphi \in C_c^\infty(\widetilde I, [0,1])$, $\varphi=1$ on $I$, and consider the function $(\widetilde f\varphi)(x)$, extended by zero to the whole real line.
		Let $\eta:\bbR\to [0,\infty)$ be a smooth function with compact support in $(-1,1)$ such that $\int_\bbR \eta(t) \, dt = 1$. Let $(\rho_n)\subset\bbR$ a sequence, to be chosen later on, such that $\rho_n\downarrow 0$ as $n\uparrow \infty$. Let $\widetilde{f}_n\coloneqq  \eta_{\rho_n} \ast  (\widetilde{f} \varphi)$ where $\eta_{\rho_n}(x)\coloneqq  \frac{1}{\rho_n}\eta(\frac{x}{\rho_n})$.
		Let us put also $\eta'_{\rho_n}(t) \coloneqq  \frac{1}{\rho_n}\eta'(\frac{t}{\rho_n})$, and notice that it belongs to $L^1(\bbR)$:
		$$ \int_\bbR |\eta'_{\rho_n}(t)| \, dt = \int_\bbR |\eta'(t)| \, dt <\infty. $$
		It is well-known by standard properties of mollification that $\widetilde{f}_n \to \widetilde{f}\varphi$ in $H^1(\bbR)$, which implies that
		$f_n \coloneqq   \widetilde{f}_n|_I\to f$ in $H^1(I)$.
		Observe now that
		$$
		\begin{aligned}
			|\widetilde{f}_n''(x)|^2 &=\frac 1 {\rho_n^2} |\int_{\bbR} \frac{1}{\rho_n}\eta'(\frac{x-z}{\rho_n}) (\widetilde f\varphi)'(z)\,dz|^2 = \frac 1 {\rho_n^2} |(\eta'_{\rho_n}\ast(\widetilde f\varphi)')(x)|^2,\\
		\end{aligned}
		$$
		from which, by Young's theorem for convolutions \cite[Corollary 2.25]{Adams2003}, we get
		$$ \begin{aligned}
			\|f_n''\|_{L^2(I)}&\leq \|\widetilde{f}_n''\|_{L^2(\bbR)}\\
			&= \frac{1}{\rho_n}\n{\eta'_{\rho_n}\ast (\widetilde f\varphi)'}_{L^2(\bbR)}\\
			&\leq \frac{C}{\rho_n} \n{\eta'_{\rho_n}}_{L^1(\bbR)}\n{(\widetilde f\varphi)'}_{L^2(\bbR)}  \leq \frac{C}{\rho_n}.
		\end{aligned}$$
		The proof is concluded by choosing $\rho_n$ such that $\e_{n}/\rho_n \to 0$ as $n \uparrow \infty$.
	\end{proof}

	\begin{lemma}\label{lemmazeta}
		Let $I\subset\mathbb{R}$ be an interval and $m \in L^\infty(I, [0,\infty))$ be a function that is constant on finitely many segments  covering $I$. Let $\gamma$ be a positive number. There exists a sequence $(\zeta_\e) \subset H^2_0(I)$ such that $\zeta_\e \weakstar 0$ in $W^{1,\infty}(I)$, $\frac12(\zeta'_\e)^2 \to m$ in $L^2(I)$, and $\e^\gamma \zeta_\e'' \to 0$ in $L^2(I)$.
	\end{lemma}
	\begin{proof}
		We can assume $I=(a,b)$ and $m$ to be a positive constant on $I$, otherwise we can perform the construction on each set on which $m$ is constant. \\
		Let  $g:\mathbb{R} \to [0,\frac12]$ be a 1-periodic function defined, on one period, by
		$$ g(t) \coloneqq  \begin{cases}
			t & 0\leq t < \frac12,\\
			1-t & \frac12<t	<1.
		\end{cases}$$
		For every $n\in\mathbb{N}$ let	$\delta_n = \sqrt{2m}(b-a)\frac1n$ and $g_{n}(t) \coloneqq  \delta_n g (\frac{\sqrt{2m} (t-a) }{\delta_n})$.
		Note that $\|g_n\|_{L^\infty(\mathbb{R} )}=\delta_n/2$ and $g_{n}'\in\{- \sqrt{2m},+\sqrt{2m}\}$ almost everywhere in $\mathbb{R}$.
		Let $(\rho_\e)$ be a sequence of positive numbers that converges to zero as $\e$ goes to zero, and let  $\eta_{\rho_\e}$ be the standard mollifier whose support is contained in $(-\rho_\e,\rho_\e)$.
		Define $\zeta_{\rho_\e,n}\coloneqq  \eta_{\rho_\e}\ast g_{n}$. It is clear that $\zeta_{\rho_\e,n} \in C^\infty(\mathbb{R})$, also $\|\zeta_{\rho_\e,n}\|_{L^\infty(\mathbb{R} )}\le\delta_n/2$ and $\|\zeta_{\rho_\e,n}'\|_{L^\infty(\mathbb{R} )}\le \sqrt{2m}$ for every $\rho_\e$ and $\delta_n$. Moreover, by Young's inequality for convolutions,
		$$\|\e^\gamma \zeta_{\rho_\e,n}''\|_{L^2(I)}=\e^\gamma \|\eta_{\rho_\e}'\ast g_{n}'\|_{L^2(I)}\le \e^\gamma
		\|\eta_{\rho_\e}'\|_{L^1(I)} \|g_{n}'\|_{L^2(I)}\le\frac{C\e^\gamma }{\rho_\e}
		$$
		and
		$$
		\begin{aligned}
			\int_I |\frac 12 (\zeta_{\rho_\e,n}')^2 -m |^2\,dx&=\int_I |\frac 12 (\zeta_{\rho_\e,n}')^2 -\frac 12 (g_{n}')^2 |^2\,dx\\
			&=\frac 14 \int_I |\zeta_{\rho_\e,n}' -g_{n}' |^2|\zeta_{\rho_\e,n}' +g_{n}' |^2\,dx\\
			&\le C \int_I |\zeta_{\rho_\e,n}' -g_{n}' |^2\,dx.
		\end{aligned}
		$$
		Hence, setting $\rho_\e=\e^{\gamma/2}$ we deduce that
		$$\lim_{\e\to 0}\|\e^\gamma \zeta_{\rho_\e,n}''\|_{L^2(I)}=0
		\quad \mbox{and}\quad
		\lim_{\e\to 0}\|\frac 12 (\zeta_{\rho_\e,n}')^2 -m \|_{L^2(I)}=0,
		$$
		for every $n$. Given that $\lim_{n\to \infty}\lim_{\e\to 0} \|\zeta_{\rho_\e,n}\|_{L^\infty(\mathbb{R} )}=0$, by a
		standard diagonalization process, we can find a map $\e \mapsto n(\e)$ such that $n(\e)\uparrow \infty$ and
		$\zeta_{\rho_\e,n(\e)} \to 0$ in $L^\infty(I)$. The proof is concluded by posing $\zeta_\e\coloneqq \zeta_{\rho_\e,n(\e)} - \zeta_{\rho_\e,n(\e)}(a)$ and by recalling the bound $\|\zeta_{\rho_\e,n}'\|_{L^\infty(\mathbb{R} )}\le \sqrt{2m}$.
	\end{proof}

	\begin{lemma}\label{lemmaa7}
		Let $I=(a,b)$ be a bounded interval. Let $\gamma > 0$ and $\alpha_a, \alpha_b \in \bbR$. Then, there exists a sequence $(g_\e) \subset H^2(I)$ such that
		\begin{equation}\label{bcg}
			g_\e(a) = g_\e(b) = 0, \quad g_\e'(a)=\alpha_a, \quad g_\e'(b)=\alpha_b \qquad \forall \e > 0\end{equation}
		and
		$$ g_\e \to 0 \quad in \ W^{1,4}(I), \qquad \e^\gamma g_\e'' \to 0 \ in \ \L^2(I). $$
	\end{lemma}
	\begin{proof}
		For simplicity, we prove the lemma for the case $I=(0,1)$ and $\alpha_b=0$. Let $\eta_\e := \e^{\gamma}$.
		Consider the sequence
		$$ g_\e(x) = \begin{cases}
			\frac{x^3}{\eta_\e^2}(\alpha_a-2) - \frac{x^2}{\eta_\e}(2\alpha_a-3) + \alpha_a x & 0<x\leq\eta_\e \\
			\frac{\eta_\e}{(1-\eta_\e)^3}(x-1)^2(2x-3\eta+1) & \eta_\e < x<1.
		\end{cases}$$
		The sequence satisfies \eqref{bcg} and
		$$ |g_\e'(x)|\leq \begin{cases}
			C_1 &  0<x\leq\eta_\e \\
			C_2 \eta_\e & \eta_\e < x<1
		\end{cases} \qquad |g_\e''(x)|\leq \begin{cases}
			\frac{C_3}{\eta_\e} &  0<x\leq\eta_\e \\
			C_4 \eta_\e & \eta_\e < x<1
		\end{cases}$$
		where $C_1, C_2, C_3, C_4$ are positive constants that may depend on $\alpha_a$.
		From these two inequalities, the lemma follows.
	\end{proof}

	\section*{Conflict of interest}
	The authors declare that they have no conflict of interest.

	\section*{Acknowledgments}
	MPS is supported by the project PRIN 2022-20229BM9EL “NutShell”.\\
	This work has been written within the activities of INdAM-GNFM.


\end{document}